\documentclass[12pt]{article}
\usepackage{amsmath,amsthm,amssymb}
\usepackage{fullpage}
\usepackage{enumerate}

\renewcommand{\pmod}[1]{\,(\textup{mod}\,#1)}
\theoremstyle{plain}

\newtheorem{theorem}{Theorem}[]
\newtheorem{lemma}[theorem]{Lemma}

\newtheorem{proposition}[theorem]{Proposition}

\newcommand{\C}{\mathbb{C}}
\newcommand{\R}{\mathbb{R}}
\newcommand{\Z}{\mathbb{Z}}
\newcommand{\SL}{\text{SL}_2(\mathbb{Z})}

\newcommand{\GN}{\Gamma}
\newcommand{\UH}{\mathcal{H}}
\newcommand{\D}{\Delta}
\def\half{{\textstyle\frac12}}
\def\sf#1#2{{\textstyle\frac{#1}{#2}}}

\makeatletter
\renewcommand{\@biblabel}[1]{#1.}
\makeatother

\begin{document}

\title{\Large \bf Parameterizations of the Chazy equation}
\author{Sarbarish Chakravarty$^1$ and Mark J Ablowitz$^2$ \\[1ex]
\small$^1$\it
Department of Mathematics, University of Colorado, Colorado Springs, CO, 80933\\ 
\small$^2$\it 
Department of Applied Mathematics, University of Colorado, Boulder, CO, 80309\\}
\date{}

\maketitle

\begin{abstract}
\noindent The Chazy equation $y''' = 2yy'' - 3y'^2$ is derived from the automorphic
properties of Schwarz triangle functions $S(\alpha, \beta, \gamma; z)$. It is shown 
that solutions $y$ which are analytic in the fundamental domain of these triangle 
functions, only correspond to certain values of $\alpha, \beta, \gamma$.
The solutions are then systematically constructed. These analytic solutions provide 
all known and one new parametrization of the Eisenstein series $P, Q, R$
introduced by Ramanujan in his modular theories of signature 2, 3, 4 and 6.
\end{abstract}


\thispagestyle{empty}
\newpage

\begin{center}
{\bf 1. Introduction}
\end{center}
In a series of papers~\cite{Chazy1}--\cite{Chazy3} between 1909 -- 1911, J. Chazy considered a 
class of nonlinear differential equations of the form
\begin{equation}
y'''- 2yy''+3y'\,^2 = \frac{4}{36-k^2}(6y'-y^2)\,^2\,, \qquad 0 \leq k \neq 6\,,
\label{chazy12}
\end{equation}
during the course of his work on the extension of Painlev\'e's program
to equations of third order. The general solutions of these equations,
which are now referred to as Chazy class XII, can be parametrized by 
solutions of appropriate hypergeometric equations for $k>0$, and the 
Airy equation when $k=0$. Furthermore, for $6 < k \in \Z$ the solutions
evolve from a generic initial data to form a natural boundary which is closed
curve in the complex plane
beyond which the solutions can not be analytically continued. For the
special cases of $k=2,3,4$ and 5, the associated hypergeometric solutions
are algebraic functions classified by Schwarz~\cite{Schwarz1}, and leads
a 3-parameter family of rational solutions for \eqref{chazy12}. 
The focus of this note is on the equation
\begin{equation}\label{chazy}
y'''- 2yy''+3y'\,^2 = 0\,,
\end{equation}
which corresponds to the limiting case $k \to \infty$ of \eqref{chazy12},
and will be referred to as the Chazy equation throughout this article.
Chazy~\cite{Chazy2} showed that the general solution of \eqref{chazy} also possesses a 
movable natural boundary by relating the solution $y(z)$ to that of the
hypergeometric equation
$s(s-1)\chi''+ (\sf{7s}{6} -\half)\chi'-\sf{1}{144}\chi=0$.

In recent years, it has been shown that the Chazy equation arises in several areas 
of mathematical physics including magnetic monopoles \cite{Atiyah-H}, 
self-dual Yang-Mills and Einstein equations \cite{CAC90, Hitchin}, and topological 
field theory \cite{Dubrovin}. In addition, \eqref{chazy} has been
derived as special reductions of hydrodynamic type equations \cite{Ferapontov}
as well as stationary, incompressible Prandtl boundary layer 
equations~\cite{Rosenhead}. These results have renewed interest in the study
of the Chazy equation. For example, the SL$_2(\C)$ symmetry of \eqref{chazy}
was exploited to systematically derive its general solution in~\cite{CO1996};
moreover, in Ref.~\cite{Ki1998} the group invariance was applied to elucidate the 
role of its pole singularities, which
under suitable perturbations coalesce to the natural boundary.
Yet another interest in Chazy equation stems from its well-known connection 
with the automorphic forms associated with the modular group 
$\SL$ and its subgroups. The notation $\GN(1) = \SL$ for the full modular group 
(cf. \cite{Rankin1}) will be used throughout this article.  
It is quite significant that the emergence of the Chazy equation \eqref{chazy}
in the theory of modular forms and elliptic functions can be traced back to the work 
of Ramanujan and subsequently in the work of Rankin and others. Here we briefly recall 
some of these interesting results which were anteceded by Chazy's work, but were 
apparently not noticed by those researchers.

In 1916, Ramanujan \cite{Ramanujan-arith}, \cite[pp 136--162]{Ramanujan-collect},
introduced the functions $P(q)$, $Q(q)$ and $R(q)$ defined for $|q|<1$ by
\begin{equation}  
P(q):=1-24\sum_{n=1}^{\infty} \frac{nq^n}{1-q^n}\,, \quad
Q(q):=1+240\sum_{n=1}^{\infty}\frac{n^3q^n}{1-q^n} \,, \quad
R(q):=1-504\sum_{n=1}^{\infty} \frac{n^5q^n}{1-q^n}\,,
\label{PQR}
\end{equation}
and proved by using trigonometric series identities that these functions in
\eqref{PQR} satisfy the differential relations
\begin{equation}
\delta P =\frac{P^2-Q}{12}\,, \qquad
\delta Q=\frac{PQ-R}{3}\,, \qquad
\delta R =\frac{PR-Q^2}{2}\,, \qquad \delta := q \frac{d}{dq}\,.
\label{dPQR}
\end{equation}
By expressing the system \eqref{dPQR} as a single differential equation
for $P(q)$, then defining 
\begin{equation}
y(z):= \pi i P(q)\,, \quad \quad \, q := e^{2 \pi iz}\,, \quad \text{Im}(z)>0 \,,
\label{yP}
\end{equation}
one easily recovers the Chazy equation \eqref{chazy}~\cite{Ablowitz-H3}.
(Note that in this case the natural boundary corresponds to the unit circle $|q|=1$
or the real axis $\text{Im}(z)=0$). Ramanujan's $P(q), Q(q), R(q)$ correspond to the 
(first three) Eisenstein series associated with the modular group $\GN(1)$, but this
modern terminology is not needed in the present context.
Ramanujan also considered the modular discriminant function
\begin{equation*}
\D(q):=\frac{Q^3(q)-R^2(q)}{1728}= \, q\prod_{n=1}^\infty(1-q^n)^{24} 
:= \, \sum_{n=1}^\infty\tau(n)q^n\,, \qquad \tau(n) \in \mathbb{Z}\,,
\end{equation*} 
he proved as well as conjectured many properties associated with the integer coefficients 
$\tau(n)$ above, which are referred to as Ramanujan's tau-functions. It follows from the
last two equations in \eqref{dPQR} that $P(q)$ is simply the logarithmic derivative
of $\D(q)$, that is
\begin{equation}
P(q) = \frac{\D'(q)}{\D(q)}\,, \qquad \mbox{or} \qquad y(z) = \frac{1}{2}\frac{\D'(z)}{\D(z)}\,.
\label{D-chazy}
\end{equation}
Equation \eqref{D-chazy} also follows by taking the logarithmic derivative of
the infinite product formula for $\D(q)$ above, and comparing it with the
$q$-expansion for $P(q)$ in \eqref{PQR}.
Rankin in \cite{Rankin}, showed using properties of modular forms that $\Delta(z)$ 
satisfies the equation 
$$ 2\D''''\D^3-10\D'''\D'\D^2-3\D''\,^2\D^2+ 24\D''\D'\,^2\D-13\D'\,^4=0\,, $$
which is homogeneous of degree $4$ (in both $\Delta$ and its derivatives).
Rankin's $\D$-equation follows from the Chazy equation \eqref{chazy} with 
$y(z)$ as in \eqref{D-chazy}.
In Ref.~\cite{Rankin}, Rankin also derives (among others) two more equations which
are equivalent to
\begin{equation*}
4Q\delta^2 Q - 5(\delta Q)^2 = 960 \D\,, \qquad
6R\delta^2 R - 7(\delta R)^2 = -3024Q \D\,,
\end{equation*}
which can be also deduced from the Ramanujan equations \eqref{dPQR}.  
Moreover, re-expressing the first equation above in terms of $P(q)$ and
its derivatives by employing \eqref{dPQR} and the definition of $\D(q)$,
yields \eqref{chazy} once more; while the second equation turns into the differential
consequence of the Chazy equation. The first equation for $Q$ above, was also obtained by 
B. van der Pol, who used it to derive the arithmetical identity~\cite{Vanderpole}
$$ \tau(n) = n^2\sigma_3(n) + 60 \sum_{m=1}^{n-1}(2n-3m)(n-3m)\sigma_3(m)\sigma_3(n-m)\,,
\quad \sigma_k(n):=\sum_{d|n}d^k \,, \,\, n \in \mathbb{N}\,,$$
relating Ramanujan's tau-functions and the sum-of-divisor function 
$\sigma_3(n)$. This and such other arithmetical identities follow from 
equating the $q$-expansions of both sides
of a differential (or polynomial) relation involving the modular functions.
We note here that using \eqref{yP} in the Chazy equation leads to the
identity
$$ n^2(n-1)\sigma_1(n) + \sum_{m=1}^{n-1}12m(5m-3n)\sigma_1(m)\sigma_1(n-m) = 0 \,,$$
that follows from the $q$-expansion of $P(q)$ in \eqref{PQR}.

Ramanujan extensively studied the properties of his modular functions $P,Q,R$ and 
established numerous identities involving them~\cite{Ramanujan-collect,Ramanujan-notebook}. 
Quite remarkably, he like Chazy,
also employed the theory of hypergeometric functions to establish an implicit
parametrization of the functions $P,Q,R$, which plays a crucial role in the proofs of
Ramanujan's modular identities. In his second notebook~\cite{Ramanujan-notebook},
Ramanujan considered the hypergeometric function 
$_2F_1(\half, \half; 1; x)$ related to the complete elliptic integral of 
the first kind by (see e.g.,~\cite{Bateman,Nehari})
$$
K(x) := \int_0^{\pi/2}\,\frac{dt}{\sqrt{1-x\sin^2t}} = 
\frac{\pi}{2} \,_2F_1(\half,\half;1;x) \,,
$$ 
and gave the following implicit parametrization of the functions $P,Q,R$
\begin{equation}
P(q) = (1-5x)\chi^2 + 12x(1-x)\chi\chi'\,, \quad
Q(q) = (1+14x+x^2)\chi^4\,, \quad R(q) = (1+x)(1-34x+x^2)\chi^6\,,
\label{Ppara}
\end{equation}
where $\chi(x):=\,_2F_1(\half, \half; 1; x)$ and the nome $q$ is defined by
$$ q = e^{-u}\,, 
\qquad u:= \pi\,\frac{_2F_1(\half,\half;1;1-x)}{_2F_1(\half,\half;1;x)}\,. $$ 
Equation \eqref{Ppara} originates from Jacobi's work on elliptic
functions~\cite{Jacobi}. The derivation of \eqref{Ppara} involves expressing 
$P,Q,R$ as logarithmic derivatives of the quotients of theta functions and 
application of the remarkable Jacobi-Ramanujan inversion formula
$$ \vartheta_3(0|q) := 1+2\sum_{n=1}^{\infty}q^{n^2} = \sqrt{\chi(x)}\,.$$
It is worth pointing out here that in Refs.~\cite{Chazy1,Chazy3}, 
Chazy had also noted the equivalence between \eqref{chazy} and a system of 
three first order equations introduced by G. Darboux~\cite{Darboux} in 1878
(see Section 2). This first order system is different from the Ramanujan
system \eqref{dPQR} even though its solutions were given by G. Halphen~\cite{Halphen} 
in terms of logarithmic derivatives of the 
null theta functions $\vartheta_2, \vartheta_3, \vartheta_4$. Each of the null theta
functions satisfies Jacobi's nonlinear equation  
$$ (\theta^2\theta''' - 15\theta\theta'\theta''+30\theta'^3)^2 =
32(\theta\theta''-3\theta'^2)^3 +\pi^2\theta^{10}(\theta\theta''-3\theta'^2)^2 = 0 \,,$$
for $\theta(z) := \vartheta_i(0|q)\,, \, i=2...4$ and where $q = e^{\pi i z}$~\cite{Jacobi1}.
Ramanujan relied heavily on the parametrization \eqref{Ppara} and the inversion formula
to develop his theory of modular equations involving functional relations
among complete elliptic integrals at different arguments. Interestingly, he also proposed 
alternative parametrizations for $P(q), Q(q), R(q)$ in terms of other hypergeometric 
functions (see Section 4), and where the appropriate nomes $q$ can be expressed by
\begin{equation}
 q_r = e^{-u_r}\,, \qquad u_r:= \frac{\pi}{\sin(\frac{\pi}{r})} \,
\frac{_2F_1(\frac1r,\frac{r-1}{r};1;1-x)}{_2F_1(\frac1r,\frac{r-1}{r};1;x)}\,,
\qquad r=2,3,4,6 \,. 
\label{signature}
\end{equation}
The index $r$ is referred to as the {\it signature} of Ramanujan's theories. The
case $r=2$ corresponds to Ramanujan's original theory of modular equations while signatures 
3,4 and 6 correspond to Ramanujan's {\it alternative} theories. He stated the results
in his alternative theories without proof~\cite{Ramanujan-notebook}. Proofs
were constructed much later~\cite{Bor91,BBG95} and is now
an important topic of research~\cite{Berndt1}. Indeed a unified framework
of Ramanujan's modular equations in different signatures based on the more contemporary
theory of modular forms, elliptic surfaces and Gauss-Manin connections was proffered
only recently in Ref~\cite{Maier1}.

In this paper, we derive the parametrization for the functions $P,Q,R$ in Ramanujan's 
original and alternative theories, as well as some new ones (see cases (1), (2) in
Table 1) via the parametrization of the Chazy solution 
$y(z)$ in terms of Schwarz triangle and hypergeometric functions. Our approach is 
not number theoretic, but rather based on the theory of Fuchsian differential equations 
and the action of its projective monodromy group on certain differential polynomials.
Section 2 provides some necessary background on Fuchsian equations and their role
in the conformal mappings of the complex plane to triangular domains T bounded by 
circular arcs. The triangle functions introduced by Schwarz
form the natural coordinates on T,
and which remain invariant under the group of automorphisms induced by the
projective monodromy groups of the Fuchsian equations. The relationship between the
triangle functions and the Chazy equation is established in Section 3. In section 4, we
give the explicit parametrization of the Chazy solution $y(z)$ in terms of the
appropriate hypergeometric functions, thereby connecting our results to those of
Ramanujan's parametrization in his theories of modular equations. Along these lines,
we also note the work in Ref.~\cite{Huber} which utilizes arguments motivated by
the Lie symmetry analysis of \eqref{dPQR} and related differential equations
to obtain various parametrizations for the functions $P, Q, R$.
The approach in this note is different from that in Ref.~\cite{Huber}.

\begin{center}
{\bf 2. Conformal mapping and triangle functions}
\end{center}
The mapping properties of Fuchsian differential equations play
a significant role in the theory of automorphic functions. In particular,
H. A. Schwarz in 1873 carried out an exhaustive study of the conformal
maps induced by ratio of solutions of hypergeometric equations and the 
so called triangle functions~\cite{Schwarz1}. This beautiful theory
has since been developed significantly and is treated in numerous 
monographs. The brief overview presented in this section closely
follow the texts~\cite{Ford} and ~\cite{Nehari}.

A second order, Fuchsian differential equation with three regular singular 
points in the complex plane can be cast into the form
\begin{equation}
u''+\frac{V(s)}{4} \, u=0 \,, \qquad \qquad 
V(s)=\frac{1-\alpha^2}{s^2}+\frac{1-\beta^2}{(s-1)^2}+
\frac{\alpha^2+\beta^2-\gamma^2-1}{s(s-1)} \,,
\label{ueqn}
\end{equation}
where $\alpha, \beta, \gamma$ are the exponent differences (for any
pair of linearly independent solutions) prescribed at the
singular points $0,\,1$ and $\infty$, respectively. (Note that $f'$ 
indicates derivation with respect to the argument
of the function $f$ throughout this article). The ratio $z(s)$ of 
any two linearly independent solutions $u_1, u_2$ of \eqref{ueqn}
is a multi-valued function branched at the regular singular points,
and satisfies the Schwarzian differential equation
\begin{equation}
\{z, s\} = \frac{V(s)}{2}\,, \qquad \qquad  
\{z, s\} := \frac{z'''}{z'}-\frac{3}{2}\left(\frac{z''}{z'}\right)^2 \,.
\label{schwarz} 
\end{equation} 
The monodromy group $G \subset \text{GL}_2(\C)$ resulting from the analytic 
extensions of the pair $(u_1, u_2)$ along all possible closed loops 
through an ordinary point $s_0$, is determined (modulo conjugation) by 
the exponent differences. $G$ acts projectively 
on the ratio $z(s)$ via fractional linear transformations
\begin{equation*}
z \to \gamma(z) := \frac{az+b}{cz+d} \,, \qquad 
\gamma = \begin{pmatrix} a & b \\ c & d\end{pmatrix} \in G\,.
\end{equation*}
Since both $\gamma$ and $\lambda \gamma$ yield the same fractional linear 
transformation for any complex $\lambda \neq 0$, the projectivized monodromy group
is the quotient group $\GN \cong G/\lambda I \subseteq \text{PSL}_2(\C)$ where 
$I$ is the $2 \times 2$ identity matrix. 
Note that both $z$ and $\gamma(z)$ satisfy the same equation 
\eqref{schwarz} due to the invariance property of the Schwarzian derivatives:\, 
$\{z,s\} = \{\gamma(z), s\}$. A special class of solutions of \eqref{schwarz}
was extensively investigated by Schwarz who considered the parameters in $V(s)$
to be real and $0 < \alpha, \beta, \gamma < 1$.
If $V(s)$ is to be further restricted such that $\alpha+\beta+\gamma<1$, then
a branch of $z(s)$ maps the upper-half $s$-plane 
(Im $s \geq 0$) onto a hyperbolic triangle T in the extended complex plane,
bounded by three circular arcs which enclose interior angles
$\alpha\pi$, $\beta\pi$, and $\gamma\pi$ at the vertices $z(0)$, $z(1)$, 
and $z(\infty)$. By the Schwarz reflection principle, the analytic extension 
of this branch to the lower-half plane across a line segment between
any two branch points, maps the lower-half plane to an adjacent triangle T$'$ 
that is the image of T under reflection across the circular arc which forms
their common boundary. Continuing this process, the complete set of branches of 
$z(s)$ maps the $s$-plane onto a Riemann surface spread over the $z$-plane
consisting of an infinite number of circular triangles obtained by
inversions across the boundaries of T and its images.
The necessary and sufficient condition that this Riemann surface is a plane
region D (being the uniform covering of non-overlapping triangles)
is that the exponent differences $\alpha, \beta, \gamma$ be either zero
or reciprocals of positive integers. In this case, the inverse $s(z)$ is
a single-valued, meromorphic, automorphic function whose automorphism 
group is the projective monodromy group $\GN$ defined above. That is,
$s(\gamma(z)) = s(z) \,,\gamma \in \GN$.
In this setting of conformal mapping, $\GN$ is a discrete 
subgroup of $PSL(2,\R)$, and turns out to be
the group of fractional linear 
transformations generated by an even number of
reflections across the boundaries of the circular triangles. More precisely,
let $r_{\alpha}, r_{\beta},r_{\gamma}$ be the reflections across the sides
opposite to the vertices $z(0)$, $z(1)$, and $z(\infty)$ of T. Then the
automorphism group is generated by the elements $R_{\alpha} = r_{\beta}r_{\gamma}$,
$R_{\beta} = r_{\alpha}r_{\gamma}$ and $R_{\gamma} = r_{\alpha}r_{\beta}$
which are rotations about each vertex by $2\pi\alpha, 2\pi\beta$ and
$2\pi\gamma$, and satisfy 
$$ R_{\alpha}^{1/\alpha} = R_{\beta}^{1/\beta} = R_{\gamma}^{1/\gamma} =
R_{\alpha}R_{\beta}R_{\gamma} = 1 \,.$$
A vertex with a nonzero (interior) angle $\pi/m, \, m \in \mathbb{Z}^{+}$
is called an elliptic fixed point of order $m$, whereas a vertex with zero angle is
called a parabolic fixed point of the group.

In its domain of existence D, the only possible singularities of $s(z)$ and its 
derivatives are located only at the vertices where $s(z)$ takes the value of 0, 1, 
or $\infty$. Due to the automorphic property, it is sufficient to consider the 
function $s(z)$ on the triangle T. Let $z=z_0$ be a vertex of the triangle T 
such that $s(z_0) = s_0 \in \{0, 1, \infty\}$ 
is a regular singular point of the Fuchsian equation \eqref{ueqn} with exponent difference 
$\mu \in \{\alpha, \beta, \gamma\}$. The behavior of the $s(z)$ near a vertex $z_0$
depends on whether $z_0$ is an elliptic or a parabolic fixed point of $\GN$.

(i)\, If $\mu = 1/m\,, m \in \mathbb{Z}^+$, then $z_0$ is an elliptic 
fixed point of order $m$. In this case, a pair of fundamental solutions of 
\eqref{ueqn} in the neighborhood of $s=s_0 \neq \infty$ are of the form
$$ u_1(s) = (s-s_0)^{(1+\mu)/2}\psi_1(s)\,, \qquad \qquad 
u_2(s) = (s-s_0)^{(1-\mu)/2}\psi_2(s) \,,$$
where $\psi_i(s),\, i=1,2$ admit convergent power series in the 
neighborhood of $s=s_0$, and $\psi_i(s_0) \neq 0$. By taking appropriate linear
combinations of the solutions $u_1$ and $u_2$, $z(s)$ is defined via 
$$z- z_0 = \frac{u_1}{u_2} = (s-s_0)^{\mu} \psi(s) \,,$$ 
where $\psi(s)$ is analytic near $s=s_0$ and $\psi(s) \neq 0$. The inverse 
function is single-valued, and given by
\begin{subequations}
\begin{equation}
 s(z) = s_0 + (z-z_0)^m\phi_1(z)\,, \quad \qquad \phi_1(z_0) \neq 0 \,,
\label{szero}
\end{equation}
where $\phi(z)$ is analytic near $z=z_0$. Thus, $s-s_0$ has 
a zero of order $m$ at $z=z_0$.
If $s_0 = \infty$, then by making the transformation $s'=1/s$ in \eqref{ueqn} one finds
in a similar manner as above that $s(z)$ has a pole of 
order $m$ at $z=z_0$. That is,
\begin{equation}
 s(z) = (z-z_0)^{-m}\phi_2(z)\,, \quad \qquad \phi_2(z_0) \neq 0 \,,
\label{spole}
\end{equation}
and $\phi_2(z)$ is analytic in the neighborhood of $z=z_0$. 
\end{subequations}
(ii)\, When $z_0$ is a parabolic vertex, $\mu = 0$. In this case, a pair of linearly
independent solutions of \eqref{ueqn} in a neighborhood of $s=s_0 \neq \infty$ 
is given by
$$ u_1(s) = (s-s_0)^{1/2}\psi_3(s)\,, \qquad \qquad
u_2(s) = [k \log(s-s_0) + \psi_4(s)]u_1(s) \,, $$
where $\psi_3(s_0) \neq 0, \, \psi_4(s_0) = 0$, $k$ is a constant, and 
both $\psi_3(s), \psi_4(s)$
admit convergent power series in the neighborhood of $s=s_0$. In terms of
$u_1, u_2$, the function $z(s)$ can be defined as
$$ 2\pi iz = \frac{u_2}{u_1} \qquad \mbox{or} \qquad 
\frac{2\pi i}{z-z_0} =  \frac{u_2}{u_1} \,,$$
depending on whether $z_0 = \infty$, or a finite vertex in the extended $z$-plane.    
From above, the inverse function $s(z)$ can be expressed as a power series
\begin{subequations}
\begin{equation} 
s(z) = s_0 + \sum_{n=1}^{\infty} c_nq^n\,, \qquad q := e^{2\pi iz/k}\, 
\left(\mbox{or} \,\, q := e^{\frac{2\pi i}{k(z-z_0)}}\right)\,, \qquad
c_n \in \C \,,
\label{szeroq}
\end{equation}
in the uniformizing variable $q$. Thus, $s(z)$ is a 
single-valued function of $z$,
holomorphic at $q=0$. If $s_0 = \infty$, then a similar analysis as
above can be carried out by making the transformation $s' = 1/s$,
and by introducing the same local uniformizer $q$ as in \eqref{szeroq}.
In this case, $s(z)$ has a pole at $q=0$, and the $q$-expansion
\begin{equation}
 s(z) = \frac{d}{q} + \sum_{n=0}^{\infty} d_nq^n\,, \qquad
d, \, d_n \in \C \,.
\label{spoleq}
\end{equation}
\end{subequations}

A pair of adjacent triangles T and T$'$ 
form the fundamental region X of the automorphism group $\GN$ whose action
on X tessellates the region D. The inverse function $s(z)$, which
maps X to the entire 
extended $s$-plane, generates the function field (over $\C$) of X. The
boundary of D in the $z$-plane, is a $\GN$-invariant circle which 
is orthogonal to (all three sides of) the triangle T and all its reflected images. 
This orthogonal circle is the set of limit points for the automorphic group $\GN$,
it is a dense set of essential singularities forming a {\it natural boundary} for the
function $s(z)$. In its domain of existence D, the only possible singularities
of $s(z)$ are poles which correspond to the vertices where $s(z) = \infty$.
Thus, $\GN$ is a Fuchsian group of the first kind (see e.g.,~\cite{Ford}, Sec. 30
for a definition) and $s(z)$
is a simple automorphic function of $\GN$. Fuchsian groups associated with 
differential equations \eqref{ueqn} with three regular singular points 
are referred to as triangle groups and $s(z) := S(\alpha, \beta, \gamma; z)$ 
are called Schwarz triangle functions.  It follows from \eqref{schwarz} that
the inverse function $s(z)$ satisfies the following third order nonlinear equation
\begin{equation}
\{s,z\} + \frac{s'^2}{2} V(s) = 0\,.
\label{ischwarz}
\end{equation}
Conversely, when the parameters $\alpha, \beta, \gamma$ in $V(s)$ are either zero
or reciprocals of positive integers, a three-parameter family of
solution of \eqref{ischwarz} is obtained as the inverse of the ratio
\begin{equation}
z(s) = \frac{A u_1(s)+B u_2(s)}{Cu_1(s) +Du_2(s)}\,, \qquad
A,B,C,D \in \C\,, \quad AD-BC =1\,,
\label{zdef}
\end{equation}
where $u_1$ and $u_2$ are linearly independent solutions of \eqref{ueqn}.
The solution is single-valued and meromorphic inside a disk in the extended
$z$-plane, and can not be continued analytically across the boundary of
the disk. This boundary is movable as its center and radius are
completely determined by the initial conditions which depend
on the complex parameters $A,B,C,D$.

A number of nonlinear differential equations whose solutions
possess movable natural boundaries, can be solved by first transforming them
into a Schwarzian equation \eqref{ischwarz} and then following the 
{\it linearization} scheme described above. We briefly recount some examples
of such nonlinear differential equations and their relations to Schwarz triangle functions. 
In 1881, Halphen considered a slightly different 
version~\cite[pp 1405, Eq (5)]{Halphen1} of the following nonlinear 
differential system 
\begin{align} 
w_1'&=-w_2w_3+w_1(w_2+w_3)+\tau^2\,,  \nonumber\\ 
w_2'&=-w_3w_1+w_2(w_3+w_1)+\tau^2\,,  \label{gdh}\\ 
w_3'&=-w_1w_2+w_3(w_1+w_2)+\tau^2\,,  \nonumber  
\end{align}
$$ \tau^2 = \alpha^2(w_1-w_2)(w_2-w_3)+\beta^2(w_2-w_1)(w_1-w_3)+
\gamma^2(w_3-w_1)(w_2-w_3) \,,$$
for functions $w_i(z) \neq w_j(z),\, i\neq j, \, i,j \in \{1,2,3\}$, and constants 
$\alpha,\, \beta,\,\gamma$. Halphen presented the solutions of
\eqref{gdh} in terms of hypergeometric functions. 
More recently, the authors found that \eqref{gdh} arises
as a symmetry reduction of self-dual Yang-Mills equations, and called it 
the generalized Darboux-Halphen (gDH) system~\cite{Ablowitz-H1,Ablowitz-H2}. 
If in fact, $w_1(z),\, w_2(z), \, w_3(z)$ are parametrized in terms of
a single function $s(z)$ (and its derivatives) as
\begin{equation}
w_1 = \frac{1}{2}\Big[\log\Big(\frac{s'}{s} \Big)\Big]'\,, \quad
w_2 = \frac{1}{2}\Big[\log\Big(\frac{s'}{s-1} \Big)\Big]'\,, \quad
w_3 = \frac{1}{2}\Big[\log\Big(\frac{s'}{s(s-1)} \Big)\Big]'\,,
\label{wdef}
\end{equation}
then $s(z)$ is a solution of \eqref{ischwarz}, where the constants
$\alpha,\, \beta,\,\gamma$ in $V(s)$ of \eqref{ischwarz} are the
same as those appearing in $\tau^2$ of \eqref{gdh}.
The special case $\alpha = \beta = \gamma = 0$ of the gDH system corresponds
to the ``classical'' Darboux-Halphen (DH) system which is equation \eqref{gdh} 
with $\tau^2=0$.
This equation originally appeared in Darboux's work of triply orthogonal 
surfaces on $\mathbb{R}^3$ in 1878 \cite{Darboux}, and its solution was 
subsequently given by Halphen \cite{Halphen} in 1881. The variables $w_i$ 
associated with the DH system can be parametrized as in \eqref{wdef} by the 
triangle function $s(z) = S(0,0,0;z)$; the latter is related to 
the elliptic modular function
$\lambda(z) := \vartheta_2^4(0|z)/\vartheta_3^4(0|z)$, expressed
in terms of null theta functions. Chazy showed that the function
$y(z):=2(w_1+w_2+w_3)$ satisfies \eqref{chazy}
introduced in Section 1. Chazy~\cite{Chazy3} also noted that besides the elliptic 
modular function $\lambda(z)$, the 
solution to \eqref{chazy} can also be given in terms of the triangle function 
$S(\half,\frac13,0;z)$ which is the same as the modular $J$-function for the group $\SL$.
In more recent work Bureau~\cite{Bureau} in 1987, investigated a class of third order nonlinear
equations, and expressed their general solutions in terms of the Schwarz 
triangle functions $s := S(\alpha,\beta,\gamma;z)$.
In particular, Bureau's class includes the Chazy equation \eqref{chazy}.

It is natural to inquire whether the solution of the Chazy equation \eqref{chazy} admits
parametrization in terms of {\it other} triangle functions besides $S(0,0,0;z)$ and
$S(\half,\frac13,0;z)$. One motivation of the present work is to address this
question and to investigate the possible linearizations
of the Chazy equation \eqref{chazy} via solutions of the Fuchsian equation \eqref{ueqn}
with parameters $\{\alpha, \beta, \gamma\}$ other than $\{\half,\frac13,0\}$
or $\{0, 0, 0\}$.
In the following we outline a method to systematically derive the Chazy equation from
the solution of the Schwarz equation \eqref{ischwarz} for appropriate values of the
parameters $\alpha, \beta, \gamma$. Our construction utilizes the transformation properties 
of $s(z)$ and its derivatives under the automorphism group $\GN$ and selects 
those triangle functions which provide a natural parametrization of the Chazy solution
$y(z)$ that is holomorphic in their domain of existence D.

\begin{center}
{\bf Triangle functions and the Chazy equation}
\end{center}
Let $\GN \subset \text{PSL}_2(\mathbb{R})$ be a Fuchsian triangle group 
with $\alpha, \beta, \gamma$ either zero or a reciprocal of positive 
integers, as in the previous section, and let $s(z)$ be a simple 
automorphic function of $\GN$
defined on a domain D of the complex plane. A meromorphic function $f$ on
D is called a {\it automorphic form} of {\it weight k} for
$\Gamma$ if
\begin{equation*}
f(\gamma(z))=(cz+d)^kf(z)\,, \qquad
\gamma = \begin{pmatrix} a & b \\ c & d\end{pmatrix} \in \Gamma\,,
\qquad \gamma(z) = \frac{az+b}{cz+d}
\end{equation*}
for all $z\in$ D. If $k=0,$ then $f$ is called a 
{\it automorphic function} on $\GN$ and is a rational function of $s(z)$.
Consider a gDH system ~\eqref{gdh} for $\GN$ where the gDH variables
$w_i,\, i=1, 2, 3$ are parametrized as in \eqref{wdef}.
From the automorphic property: $s(\gamma(z)) = s(z)$, it
follows that $s'(\gamma(z)) = (cz+d)^2s'(z)$, and that
\begin{equation*}
w_i(\gamma(z)) = (cz+d)^2w_i(z)+ c(cz+d)\,, \qquad \gamma \in \GN \,.
\end{equation*}
That is, $s'(z)$ is a weight 2 automorphic form, whereas the $w_i$ are called
{\it quasi}-automorphic forms of weight 2. Define in terms of the gDH variables,
the following function
\begin{equation}
y(z) = a_1w_1 + a_2w_2 + a_3w_3\,,
\label{ydef}
\end{equation}
on D, where the coefficients $a_i$ are constants. The objective of this section
is to determine an autonomous differential equation which is a polynomial
in $y$ and its derivatives.

It follows from the transformation
property of the $w_i$ above, that $y(z)$ transforms under the action of 
$\GN$ as
\begin{equation}
y(\gamma(z)) = (cz+d)^2y(z)+ pc(cz+d)\,, \qquad \gamma \in \GN \,,
\label{ytransf}
\end{equation}
where $p = a_1+a_2+a_3$ is called the coefficient of {\it affinity} 
of the quasi-automorphic form $y(z)$. 
Furthermore, a sequence of automorphic forms
can be constructed on the ring of differential polynomials of $y(z)$
as follows:
\begin{lemma} \label{fk}
Let $y(z)$ be a quasi-automorphic form of $\GN$ with affinity coefficient $p$, 
then $f_2 = y'- y^2/p$ is an automorphic form of weight 4, and
$f_{n+1} = f_n' - (2n/p)yf_n, \, n \geq 2$ are automorphic forms of weight $2n+2$
of $\GN$.
\end{lemma}
\begin{proof}
Differentiating \eqref{ytransf}, one finds that  
$$y'(\gamma(z))=
(cz+d)^2\left((cz+d)^2y(z)+ pc(cz+d)\right)' = 
(cz+d)^4\left(y'(z) + \frac{2cy(z)}{cz+d} + \frac{pc^2}{(cz+d)^2}\right)$$
and also from \eqref{ytransf},
$$ y(\gamma(z))^2= (cz+d)^4\left(y(z)^2 + \frac{2pcy(z)}{cz+d}+
\frac{p^2c^2}{(cz+d)^2} \right) \,. $$
Then, by combining the two expressions above yields the desired transformation 
property for $f_2$. The rest follows by induction and use of \eqref{ytransf}. 
\end{proof}
The parametrization \eqref{wdef} of the gDH variables and \eqref{ydef}
yield the following expression for $y(z)$ in terms of $s(z)$ and its derivatives
\begin{equation}
y(z) = \frac{p}{2}\, \frac{\phi'(z)}{\phi(z)}\,, \quad \qquad 
\phi(z) = \frac{s'(z)}{(s-1)^{b_1}s^{b_2}}\,,
\label{ypara}
\end{equation}
with $b_1= 1-a_1/p$ and $b_2= 1-a_2/p$. Then all the higher
derivatives of $y$ can be expressed also in terms of $s(z), s'(z)$ and 
$s''(z)$ by differentiating \eqref{ypara} successively  and using the 
Schwarz equation \eqref{ischwarz} for $s(z)$. In particular, the 
differential polynomials $f_k$ introduced in Lemma~\ref{fk} are given by
\begin{equation}
f_n = -\frac{p}{2}\,V_n(s)s'(z)^n\,, \qquad V_{n+1}(s) = V_n'(s) + nq(s)V_n(s)\,,
\quad n\geq 2\,,
\label{Vk}
\end{equation}
where $V_n(s)$ are rational functions of $s$ defined recursively from
\begin{equation}
 V_2(s) = \half V(s) + q'(s) + \half q(s)^2 \,, \qquad
q(s) = \frac{b_1}{s-1}+\frac{b_2}{s}\,,
\label{V2}
\end{equation}
and where $V(s)$ is given in \eqref{ueqn}. Combining Lemma~\ref{fk} with
\eqref{Vk} for $n=2, 3, 4$, and with $f_2 \neq 0$, one can construct 
the following rational expressions in $y(z), y'(z), y''(z)$ and $y'''(z)$ 
\begin{equation}
 \frac{f_3^2}{f_2^3} = -\frac{2}{p}\frac{V_3(s)^2}{V_2(s)^3}\,, \qquad
\frac{f_4}{f_2^2} = -\frac{2}{p}\frac{V_4(s)}{V_2(s)^2}\,,
\label{I}
\end{equation}
which are rational in $s(z)$. Eliminating $s$ from the two equations
in \eqref{I} leads, in the general case, to a third order, nonlinear
equation, rational
in $y(z)$ and its derivatives, and which depends explicitly on the 
parameters $p, b_1, b_2$ (equivalently, $a_1, a_2, a_3$),
 and $\alpha, \beta, \gamma$.  
For suitable choices of these parameters, the third order equation 
constructed this way can be identified with a large set of nonlinear 
differential equations, and whose solutions $y(z)$ are given via 
\eqref{ypara}, in terms of the solutions of the Schwarzian equation 
\eqref{ischwarz}. Furthermore, it follows from Section 2 that all such
solutions will admit a natural boundary, and will be meromorphic in the
domain of existence D.  Of these third order equations,
only the special case of the Chazy equation will be considered
in this note, leaving the general classification problem for a future work. 

The Chazy equation \eqref{chazy} can be expressed simply as a 
polynomial equation in terms of the automorphic forms $f_2$ and $f_4$.
Indeed, from Lemma~\ref{fk}
$$ f_4 = y'''-\frac{12}{p}yy''+\frac{18}{p}y'^2 -\frac{24}{p}f_2^2\,.$$ 
Then imposing the constraint
$$ p := a_1+a_2+a_3 = 6 \,, $$ 
on the coefficients $a_i$ in \eqref{ydef} in the above 
expression for $f_4$, yields the alternative expression:\,
$f_4 + 4f_2^2 = 0$ for the Chazy equation \eqref{chazy}. Note that this
expression is equivalent to the vanishing of a certain automorphic form
of weight 8 associated with the Fuchsian group $\GN$ (see Section 4). 
Then the second equation in \eqref{I} implies that the Chazy equation must 
be equivalent to the condition
\begin{equation}
 V_4 = 12 V_2^2 \,,
 \label{chazy1}
\end{equation}
which needs to hold for {\it all} $s$. Condition \eqref{chazy1} imposes 
certain restrictions on the parameters $b_1, b_2$, and 
$\alpha, \beta, \gamma$ appearing in $V_2(s), V_4(s)$.
However, we impose further constraints on the parameters $b_1, b_2$ by
demanding that the meromorphic function $y(z)$ be in fact, {\it holomorphic}
in its domain of existence D. The reason for this additional condition 
is motivated from the known result (see e.g.~\cite{Chazy2}) that the general 
solution $y(z)$ of the Chazy equation obtained via the triangle function
$S(\half, \sf13, 0; z)$ is analytic on D although there exists particular 
(2-parameter family) solutions that are meromorphic but do not possess a natural
boundary. Since one of the main objectives of this paper is to obtain
parametrizations of the Chazy equation in terms of triangle functions,
we impose the holomorphicity of $y(z)$ a priori; then this leads to the specific 
choices for the parameters $\{\alpha, \beta, \gamma\}$ for the
triangle functions which parametrize 
the Chazy solution $y(z)$. These algebraic conditions will be systematically
investigated next.

In order to determine whether $y(z)$ defined in \eqref{ydef} is
holomorphic on D, it is necessary to analyze the singularities of the
gDH variables $w_i$ given in terms of $s(z)$ and its derivatives 
in \eqref{wdef}. It follows from the conformal mapping theory 
discussed in Section 2 that it is sufficient to examine the behavior
of the function $s(z)$ and its derivatives near
the vertices $z(0), z(1)$, and $z(\infty)$ of the fundamental
triangle T. The Schwarz reflection principle and the automorphic property then ensure 
that $s(z)$ will have the same behavior at the vertices of the reflected triangles in D.
\begin{lemma} \label{wi}
Let $s(z) = S(\alpha, \beta, \gamma; z)$ be the Schwarz triangle function of 
a Fuchsian group $\GN$ with fundamental triangle T whose interior angles at the
vertices $z(0), z(1), z(\infty)$ are respectively, $\alpha \pi, \beta \pi$ and
$\gamma \pi$. Furthermore, let the gDH variables $w_1, w_2, w_3$ be given 
by $s(z)$ and its derivatives as in \eqref{wdef}.
If $\{\alpha, \beta, \gamma\} = \{1/p_1, 1/p_2, 1/p_3\}, \, p_j \in \mathbb{Z}^+, 
\, j=1, 2, 3$, then $w_1, w_2, w_3$ have first order poles at each 
of the vertices with the following residue scheme:  
\begin{gather*}
\mathrm{Res}_{z=z(0)}\{w_1, w_2, w_3\} = \{-\half, \half(p_1-1), -\half\}\,, \qquad
\mathrm{Res}_{z=z(1)}\{w_1, w_2, w_3\} = \{\half(p_2-1), -\half, -\half\}\,, \\
\mathrm{Res}_{z=z(\infty)}\{w_1, w_2, w_3\} = \{-\half, -\half, \half(p_3-1)\}.
\end{gather*}
\end{lemma}
\begin{proof}
Direct computation of $w_1, w_2, w_3$ in \eqref{wdef}
using equations \eqref{szero} and \eqref{spole}.
\end{proof}
It follows from Lemma \ref{wi} that only possible singularities of  
$y(z)$ defined by \eqref{ydef} in the domain D, are first order poles at the 
vertices corresponding to the elliptic fixed points of $\GN$. (Note however that 
$y(z)$ has a dense set essential singularities at the boundary of D as is
the case for the function $s(z)$). Therefore, $y(z)$ will be holomorphic in the entire
domain D if the coefficients $a_i$ in \eqref{ydef} can be chosen such that the
residue at each pole vanishes. Consider first the case $\{\alpha, \beta, \gamma\} = 
\{1/p_1, 1/p_2, 1/p_3\}$, where $p_j \in \mathbb{Z}^+$. From part (i)
of Lemma \ref{wi}, the condition that the residue of $y(z)$ vanishes at each pole
$z(0), z(1), z(\infty)$, leads to a set of homogeneous, linear equations for $a_i$, namely,
\begin{align}
a_1 + (1-p_1)a_2+a_3 &= 0\,,  \nonumber\\
(1-p_2)a_1+a_2+a_3 &= 0\,,  \label{ai}\\
a_1+a_2+(1-p_3)a_3 &= 0\,.  \nonumber
\end{align}
For nontrivial solutions $a_i$, the vanishing of the determinant of the coefficient 
matrix in \eqref{ai} implies that $1/p_1+1/p_2+1/p_3 = \alpha + \beta + \gamma = 1$,
while the conformal mapping of the upper-half $s$-plane onto the hyperbolic
triangle requires that $\alpha + \beta + \gamma < 1$. 
Therefore, only two out of three equations in \eqref{ai} can be used to set
the residue equal to zero. Thus, for $y(z)$ to be holomorphic on D, one of the
three vertices must be a parabolic vertex at the boundary of D instead
of an elliptic vertex with a pole singularity. Hence, we have the following result.
\begin{proposition} \label{holo}
If a solution $y(z)$ of the Chazy equation \eqref{chazy} is parametrized by the 
triangle function $s(z)=S(\alpha, \beta, \gamma; z)$ defined on a domain D, then a necessary 
condition for $y(z)$ to be holomorphic on D is that at least one of the
parameters $\alpha, \beta, \gamma$ be zero.
\end{proposition} 
Here we emphasize that it is indeed possible for the solutions $y(z)$ of 
other nonlinear equations obtained from the general construction outlined 
below \eqref{I}, to be meromorphic
in its domain of existence D. For example, the general solutions of \eqref{chazy12}
for $6 < k \in \mathbb{Z}$ are given in terms of the triangle functions 
$S(\half, \sf{1}{3}, \sf{1}{k}; z)$, and are meromorphic in D with poles
at the elliptic vertices which have interior angles $\pi/k$~\cite{Chazy2}. 
We will consider
the parametrization of \eqref{chazy12} and other equations in a future study.

According to Proposition \ref{holo}, at least one of the vertices of the fundamental
triangle T must be a parabolic fixed point of the automorphism group $\GN$.
Thus there are three distinct cases which are considered below.
For reasons that will be clear in Section 4, $z(0)$ is always chosen to be a
parabolic vertex in each of these cases so that the corresponding triangle 
functions are of the form $S(0, \beta, \gamma; z)$.

{\it 2 elliptic and 1 parabolic vertices}:\, As indicated above, 
it suffices to choose  
$z(0)$ to be the parabolic vertex, and $z(1), z(\infty)$ as the elliptic vertices 
so that $\{\alpha, \beta, \gamma\} = \{0, 1/p_2, 1/p_3\}$. All other sub cases 
can be generated from permutations of $\{z(0), z(1), z(\infty)\}$.
Requiring the residue of $y(z)$ to be zero at each of the poles $z(1), z(\infty)$ yields
the last two equations of \eqref{ai}. From these, one easily deduces that
$a_1 = p/p_2 = \beta p$ and $a_3 = p/p_3 = \gamma p$ so that $a_2/p = 1 - \beta - \gamma$.
Moreover, the parameters $b_1, b_2$
in \eqref{ypara} are given by $b_1 = 1-a_1/p = 1-\beta$ and $b_2 = 1-a_2/p= \beta + \gamma$.
Equations \eqref{V2} and \eqref{Vk} then deliver the explicit forms for $V_2(s)$
and $V_4(s)$, which depend on the remaining parameters $\beta, \gamma$. They are given by
\begin{gather*}
V_2 = \frac{A}{s^2(s-1)}\,, \qquad \mbox{with} 
\qquad A = \frac{-\,(1-\beta-\gamma)^2}{2}\, \qquad \quad \mbox{and}\,, \\
V_4 = A\,\frac{(1-2\gamma)(1-3\gamma)s^2 + 
[(1-\beta-\gamma)(12\gamma-7)+(1-2\gamma)]s + 6(1-\beta-\gamma)^2}{s^4(s-1)^3}\,.
\end{gather*}
Finally, inserting the above expressions in 
\eqref{chazy1} and requiring that the resulting expression 
$$ s[(1-2\gamma)(1-3\gamma)s +(1-\beta-\gamma)(6\gamma-6\beta-1) + (1-2\gamma)] = 0 $$
holds for all $s \notin \{0, 1, \infty\}$, give rise to two distinct solutions: 
$$ \{\alpha, \beta, \gamma\} = \{0, \half, \sf13\}\,, \qquad 
\{\alpha, \beta, \gamma\} = \{0, \sf13, \sf13\}\,.$$

{\it 1 elliptic and 2 parabolic vertices}:\, Again, without any loss of generality,
the elliptic vertex can be chosen as $z(1)$ such that 
$\{\alpha, \beta, \gamma\} = \{0, 1/p_2, 0\}$. From
Lemma \ref{wi}, $y(z)$ is holomorphic in T except at the vertex $z(1)$
where it has a simple pole. Vanishing of the residue at the pole $z(1)$ 
gives the second equation in \eqref{ai}, which implies that
$a_1 = p/p_2 = \beta p$ and $b_1 = 1-a_1/p= 1-\beta$. In this case,
$b_2$ and $\beta$ are the two remaining parameters in $V_2(s)$ and $V_4(s)$,
which take the forms
 $$ V_2 = \frac{(1-b_2)^2}{2s^2}+ \frac{B}{s(s-1)}\,, \qquad 
V_4 = \frac{N(s)}{s^4(s-1)^3}\,,$$
where $B:=(1-\beta)(2b_2-\beta-1)/2$ and $N(s)$ is a cubic polynomial
whose coefficients depend on $\beta, b_2$. Then 
\eqref{chazy1} implies that $C_0 - C_1s(s-1)=0$ for all 
$s \notin \{0, 1, \infty\}$, where  
$$ C_0 = \frac{(1-\beta)(1-2\beta)(1-3\beta)}{2}\,, \quad
C_1 = (1-\beta)(2-3\beta)(1-b_2)^2+B[6(1+\beta^2)-5(b_2+\beta)]\,,$$
and $B$ is defined above. Setting $C_0 = 0 = C_1$, three distinct solutions are found:
\begin{gather*}
(i)\, \{\alpha, \beta, \gamma\}=\{0, \sf12, 0\}\,, \quad b_2 \in \{\sf23, \sf56\};
\qquad (ii)\, \{\alpha, \beta, \gamma\} = \{0, \sf13, 0\}\,, 
\quad b_2 \in \{\sf12, \sf56\}; \\
(iii)\, \{\alpha, \beta, \gamma\} = \{0, \sf23, 0\}\,, \quad b_2 = \sf56 \,.
\end{gather*}
Notice that in case (iii), $\beta$ is {\it not} reciprocal of a positive integer,
hence $S(0, \sf23, 0; z)$ is not a single valued function of $z$.
 
{\it 3 parabolic vertices}:\, Here, 
$\{\alpha, \beta, \gamma\} = \{0, 0, 0\}$ so that $y(z)$ is holomorphic on D
for any choice of the parameters $a_i$. The functions $V_2(s)$ and $V_4(s)$
depend on the parameters $b_1, b_2$ which are determined from \eqref{chazy1}
by requiring that $y(z)$ satisfies the Chazy equation \eqref{chazy}.
The calculations are tedious but similar to the previous two cases, and there are 
two distinct solutions: $b_1 = b_2 = \sf23$, and $b_1 = b_2 = \sf56$.

The triangle functions $S(\alpha, \beta, \gamma; z)$ corresponding to the three cases
considered above are the ones which give holomorphic solutions of the 
Chazy equation \eqref{chazy} via \eqref{wdef} and \eqref{ydef}. These results are 
summarized in Table 1 below. Recall that $a_1+a_2+a_3 := p =6$, and that
$y = 3\phi'(z)/\phi(z)$ as given in \eqref{ypara}. The automorphic groups $\GN$ 
corresponding to the triangle functions are listed in the second column of the table.
With the exception of Case 5, all others are subgroups of the modular group 
$\GN(1)$ (see e.g., \cite{Rankin1} for notation);
these will be discussed in the following section. The triangle function in Case 5
is not a simple automorphic function of a Fuchsian group of first kind since 
the exponent difference at the vertex $z(0)$
given by $\beta = \frac23$ is {\it not} reciprocal of a positive integer.
Hence $S(0, \frac23, 0; z)$ is not a single valued function of $z$ on the domain D.
Nevertheless, one can show from the conformal mapping properties of the triangle 
functions that the single-valued function $S(0, \frac12, \frac13; z)$
can be expressed as a degree-2 rational function of $S(0, \frac23, 0; z)$~\cite{Nehari}, 
namely
$$ S(0, \sf12, \sf13; \epsilon z) = 
\frac{-4S(0, \sf23, 0; z)}{[S(0, \sf23, 0; z)-1]^2} \,, $$
with $\epsilon = \sqrt[3]{-\sf14}$. Then $\phi(z)$ in case 5 is a constant multiple 
of the $\phi(\epsilon z)$ in case 1; 
hence, $y(z)$ obtained from $S(0, \sf23, 0; z)$ is a single-valued function.

\begin{table}[h]
\begin{center}
\begin{tabular}{|c|c|c|c|c|} \hline
Case & $\GN$ & $s(z) = S(\alpha, \beta, \gamma; z)$ & $y=a_1w_1+a_2w_2+a_3w_3$ &
$\phi(z)=s'(z)/(s-1)^{b_1}s^{b_2}$ \\ \hline
1 & $\GN(1)$ & $S(0, \half, \sf13; z)$ & $y=3w_1+w_2+2w_3$ &
$\phi = s'/(s-1)^{1/2}s^{5/6}$ \\ \hline
2 & $\GN^2$ & $S(0, \sf13, \sf13; z)$ & $y=2w_1+2w_2+2w_3$ &
$\phi = s'/(s-1)^{2/3}s^{2/3}$ \\ \hline
3 & $\GN_0(2)$ & $S(0, \half, 0; z)$ &\begin{tabular}{c}
(i)\, $y=3w_1+2w_2+w_3$ \\
(ii)\, $y=3w_1+w_2+2w_3$
\end{tabular}
&\begin{tabular}{c}
(i)\, $\phi = s'/(s-1)^{1/2}s^{2/3}$ \\
(ii)\, $\phi = s'/(s-1)^{1/2}s^{5/6}$
\end{tabular} \\ \hline
4 & $\GN_0(3)$ & $S(0, \frac13, 0; z)$ &\begin{tabular}{c}
(i)\, $y=2w_1+3w_2+w_3$ \\
(ii)\, $y=2w_1+w_2+3w_3$
\end{tabular}
&\begin{tabular}{c}
(i)\, $\phi = s'/(s-1)^{2/3}s^{1/2}$ \\
(ii)\, $\phi = s'/(s-1)^{2/3}s^{5/6}$
\end{tabular} \\ \hline
5 &  & $S(0, \frac23, 0; z)$ & $y=4w_1+w_2+w_3$ &
$\phi = s'/(s-1)^{1/3}s^{5/6}$ \\ \hline
6a & $\GN_0(4)$ & $S(0, 0, 0; z)$ & $y=w_1+w_2+4w_3$ &
$\phi = s'/(s-1)^{5/6}s^{5/6}$ \\ \hline
6b & $\GN(2)$ & $S(0, 0, 0; z/2)$ & $y=2(w_1+w_2+w_3)$ &
$\phi = s'/(s-1)^{2/3}s^{2/3}$ \\ \hline
\end{tabular}
\end{center}
\kern-\medskipamount
\caption{Triangle functions associated with the Chazy solution $y(z)$}
\label{t:chazy}
\end{table}
Note that in each of the cases 3 and 4, the two different forms of $y(z)$ can be transformed
to each other by interchanging $w_2$ and $w_3$. This transformation stems from the permutation 
of the two parabolic vertices $z(0)$ and $z(\infty)$ of the fundamental triangle T, thereby
inducing the inversion map $s \to s^{-1}$ of the function field $s(z)$. Under this
involution, $\phi(z)$ corresponding to the sub cases (i) and (ii) are constant multiple of 
each other, for both cases 3 and 4. In case 6, there are two automorphism groups which are 
conjugate to each other:\, $\GN_0(4) = g\GN(2)g^{-1}$ where 
$g = \big(\begin{smallmatrix} 1 & 0 \\ 0 & 2\end{smallmatrix}\big)$, and
the fundamental domain of $\GN(2)$ is mapped under $z \to 2z$ to the fundamental domain 
of $\GN_0(4)$. The canonical automorphic function for $\GN(2)$ is the elliptic modular function 
$\lambda(z)$, which is related to the triangle function of $\GN_0(4)$ as 
$ S(0, 0, 0; z) = \lambda(2z)[\lambda(2z)-1]^{-1}$.
From the above relation, it turns out that the $\phi(z)$ in case 6b can be transformed 
to that in case 6a due to the well-known (see e.g., \cite{Nehari}) functional relation:\,
$\lambda(z) = 4\sqrt{\lambda(2z)}[1+\sqrt{\lambda(2z)}]^{-2}$.
Thus we see that there are many different representations of the solution of the 
Chazy equation in terms of Schwarzian triangle functions which in turn satisfy 
Schwarzian equations \eqref{ischwarz}. As we have seen, each of these equations 
can be linearized via \eqref{ueqn}. We turn to this topic next.

\begin{center}
{\bf 4. Parametrization of the Chazy solutions} 
\end{center}
In the previous section, explicit solutions of the Chazy equation were
presented in terms of the triangle functions listed in Table 1.
The standard expressions of the triangle functions are usually via 
the theta or Dedekind's eta functions which admit
Fourier or $q$-series expansions (as in \eqref{szeroq} or \eqref{spoleq}) 
developed in the neighborhood of the parabolic vertex $z_o=i\infty$ 
(see e.g., ~\cite{Maier1} and references therein). 
It is however more convenient to express the Chazy solution $y(s(z))$ 
{\it implicitly}, that is, in 
terms of the variable $s$ and a solution $u(s)$ of the linear equation \eqref{ueqn}.
It is then possible to treat the nonlinear Chazy equation purely on the basis of
the classical theory of linear Fuchsian differential equations. This is the main purpose
of the present section. Yet another motivation is to relate our results to Ramanujan's
work on the various parametrization of his functions $P, Q, R$, as mentioned in
the introduction.

In this section, the domain of the automorphic functions $s(z)$ will be taken
as D $= \UH$, the upper-half complex plane, and the hypergeometric form of 
the Fuchsian differential equation \eqref{ueqn} will be considered, in order 
to make contact with standard literature.
If $u(s)$ is a solution of \eqref{ueqn}, then the function
\begin{equation}
 \chi(s) = s^{(\alpha-1)/2}(s-1)^{(\beta-1)/2}u(s)
\label{chi}
\end{equation}
satisfies the hypergeometric equation
\begin{subequations}
\begin{equation}
\chi''+\Big(\frac{1-\alpha}{s}+\frac{1-\beta}{s-1}\Big)\chi'+ 
\frac{(\alpha+\beta-1)^2-\gamma^2}{4s(s-1)}\chi=0 \,,
\label{hyper-1}
\end{equation}
which can be written in more standard form as  
\begin{equation}
s(s-1)\chi''+[(a+b+1)s-c]\chi'+ab\chi=0,
\label{hyper-2}
\end{equation}
\end{subequations}
where $a=(1-\alpha-\beta-\gamma)/2$, $b=(1-\alpha-\beta+\gamma)/2$,
and $c=1-\alpha$. The transformation \eqref{chi} sets one of the exponents 
to 0 at each of the singular points $s=0, 1$, but the exponent differences
as well as the ratio $z(s)$ of any two linearly independent solutions, remain
the same as those in \eqref{ueqn}. Consequently,
the conformal mapping and the construction of the triangle function described in Section 2
can be carried out in an identical manner by employing the classical theory of
the hypergeometric equation instead of \eqref{ueqn}. 
Let $\{\chi_1, \chi_2\}$ be a pair of linearly independent solutions of \eqref{hyper-1} 
or \eqref{hyper-2}, and set $z(s) = \chi_2(s)/\chi_1(s)$, then 
$s'(z) = 1/z'(s)=\chi_1^2/W(\chi_1, \chi_2)$
where $W(\chi_1, \chi_2) = Cs^{\alpha-1}(s-1)^{\beta-1}$ is the Wronskian, and $C \neq 0$ is
a constant depending on the chosen pair of solutions $\{\chi_1, \chi_2\}$.
Use of this expression for $s'(z)$ in \eqref{ypara} provides an implicit parametrization 
for $y(z)$ in terms of $\chi_1$ and its $s$-derivative, given by
\begin{equation}
y(z(s)) = \frac3C s^{-\alpha}(s-1)^{-\beta}
\Big(2s(s-1)\chi_1\chi_1'-[(\tilde{b}_1+\tilde{b}_2)s-\tilde{b}_2]\chi_1^2\Big) \,, 
\label{yimp}
\end{equation}
with $\tilde{b}_1=b_1+\beta-1$ and $ \tilde{b}_2 = b_2+\alpha-1$. 
In the following, we construct the triangle functions $s(z)$ for all the
cases listed in Table 1 as well as the corresponding Chazy solution $y(z)$ using
\eqref{yimp}. In each case, the pair of hypergeometric solutions $\{\chi_1, \chi_2\}$  
are so chosen that the conformal map results in a fundamental region T which
has the parabolic vertex $z(0) = i\infty \in \UH$, around which suitable $q$-expansions
for $s(z)$ are developed. In addition, one imposes boundary conditions on $y, y', y''$
in a consistent manner at this vertex in order to uniquely fix a special solution $y(z)$,
which is then shown to be related to Ramanujan's $P(q)$ via \eqref{yP}.
The ensuing results then
relate to Ramanujan's parametrization for his theories of signatures 2,3,4, and 6,
as mentioned in Section 1. To the best of our knowledge the cases from Table 1
that were not recorded by Ramanujan, correspond to Case 1 which was known by
Chazy himself ~\cite{Chazy2,Chazy3}, and Case 2. The representation of $y(z)$
via the triangle function $S(0, \sf13, \sf13; z)$ in case 2 can be found in 
Ref.~\cite{Takhtajan}.

\begin{center}
{\it 4.1. Case 1}
\end{center}
Here $\alpha = 0, \beta = \half$ and $\gamma = \sf13$.
The corresponding group of automorphisms, $\GN$, is generated by rotations about
the vertices $z(1)$ and $z(\infty)$ by $\pi$ and $2\pi/3$, respectively, and
a parabolic transformation stabilizing the vertex $z(0)$. It is known 
(see e.g.~\cite{Ford}) that the projective action of $\GN$ is isomorphic to 
that of the modular group 
$\GN(1) := \SL$ acting on the upper-half plane $\UH$ via fractional linear 
transformations according to
\begin{equation*}
z \to \gamma(z) = \frac{az+b}{cz+d} \,, \qquad
\gamma = \begin{pmatrix} a & b \\ c & d\end{pmatrix} \in \SL\,, 
\end{equation*}
and which is generated by the fundamental transformations:
$$ z \to z+1 \quad \mbox{(translation)}\,, \qquad z \to \frac{-1}{z} \quad 
\mbox{(inversion)}\,.$$
It is also customary to choose the fundamental triangle T in $\UH$ such
that it is a strip parallel to the imaginary axis, bounded by the unit circle 
$|z|=1$, and the vertices are located at $z(0) = i\infty,\, z(1)=i,\,
z(\infty) = e^{2 \pi i/3}$.
To construct the triangle function $S(0,\half,\sf13;z)$ on T, one starts with
the hypergeometric equation (cf. \eqref{hyper-1}, \eqref{hyper-2})
\begin{equation}
\chi'' + \left(\frac{1}{s}+\frac{1}{2(s-1)}\right)\chi' + 
\frac{5/144}{s(s-1)}\chi = 0 \,,
\label{hyper-case1}
\end{equation}
with parameters $a=\sf{1}{12}, \, b=\sf{5}{12}$ and $c=1$. Equation 
\eqref{hyper-case1} admits a one-dimensional space of single-valued 
solutions spanned by the hypergeometric function 
$_2F_1(\sf{1}{12}, \sf{5}{12};1;s)$ that is 
holomorphic in a neighborhood of $s=0$ and is normalized to unity there.
The second independent solution of \eqref{hyper-case1} is chosen
as $_2F_1(\sf{1}{12}, \sf{5}{12};\half;1-s)$, which is holomorphic in a 
neighborhood of $s=1$. Next, define the conformal map $z(s)$ as follows
\begin{equation}
z(s)= \frac{\chi_2}{\chi_1}\,, \qquad \chi_1 =\, _2F_1(\sf{1}{12}, \sf{5}{12};1;s)\,,
\quad \chi_2 = A\,\,_2F_1(\sf{1}{12}, \sf{5}{12};\half;1-s) + B\chi_1\,,
\label{z-case1}
\end{equation}
where $A, B$ are constants to be determined by fixing the vertices
of T as specified above. The analytic
continuation of $_2F_1(\sf{1}{12}, \sf{5}{12};\half;1-s)$ into the neighborhood 
of $s=0$ is given by (see e.g.~\cite{Bateman}),
$$ _2F_1(\sf{1}{12}, \sf{5}{12};\half;1-s) = \frac{\GN(\half)}{\GN(\sf{1}{12})\GN(\sf{5}{12})}
\left[-\log(s)_2F_1(\sf{1}{12},\sf{5}{12};1;s) + 
\sum_{n=0}^{\infty}\frac{(\sf{1}{12})_n(\sf{5}{12})_n}{(n!)^2}\,h_ns^n \right]\,,
$$
where $\GN(\cdot)$ are Gamma functions, $h_n:=2\psi(1+n)-\psi(\sf{1}{12}+n)-\psi(\sf{5}{12}+n)$,
and $\psi(\cdot) := \GN'(\cdot)/\GN(\cdot)$ are the Digamma functions. 
Hence from \eqref{z-case1},
$$ z(s)=A\,\frac{_2F_1(\sf{1}{12},\sf{5}{12};\half;1-s)}
{_2F_1(\sf{1}{12},\sf{5}{12};1;s)}+B
= A_1(-\log(s) + h(s)) + B\,, \qquad 
A_1 = A\,\frac{\GN(\half)}{\GN(\sf{1}{12})\GN(\sf{5}{12})}\,, $$
near $s=0$, and $h(s)$ is holomorphic in that neighborhood. The constant $A$ is 
determined by fixing $A_1 = i/2 \pi$ above, so that $z(s) \to i\infty$ as $s \to 0^+$, 
and $z \to z+1$ onto the next branch as $s$ makes a circuit around $s=0$. The constant 
$B$ can then be determined by demanding that $\displaystyle\lim_{s \to 1}z(s) = i$. 
This sets the elliptic vertex $z(1) = i$, and yields $B=-i$ from the expression
of $z(s)$ above, after using $_2F_1(\sf{1}{12},\sf{5}{12};1;1) = 
\GN(\half)[\GN(\sf{7}{12})\GN(\sf{11}{12})]^{-1}$ and the formula 
$\GN(x)\GN(1-x) = \pi \csc(\pi x)$.
Consequently, the function $q := e^{2 \pi iz}$ has the power series representation
$ q = se^{-h(s)+ 2\pi} = B_1 s (1 + a_1s + a_2s^2 + \ldots)$, which can be inverted
to obtain a $q$-series representation of the triangle function 
$s(z)=S(0,\sf12, \sf13; z)$
in the neighborhood of the parabolic vertex $q=0 \,(z = i\infty)$.
In particular, the constant $B_1 = e^{2\pi-h(0)}$ can be evaluated from the
the analytic continuation formula of $ _2F_1(\sf{1}{12}, \sf{5}{12};\half;1-s)$ 
given above. One easily checks that 
$h(0) = h_0 := 2\psi(1) - \psi(\sf{1}{12}) - \psi(\sf{5}{12})$ which
evaluates to $h_0 = 2 \pi + 3\log 12$ using Gauss's Digamma 
formula (see~\cite{Bateman}), and thus $B_1 = 1/1728$. The first few terms of the 
$q$-expansion of the triangle function is then given by
$$ S(0,\sf12, \sf13; z) = 1728q(1-744 q +356652 q^2 + \ldots)\,.$$
Using the analytic continuations of the pair $\{\chi_1, \chi_2\}$
to a neighborhood of $s=\infty$, and from the values of the constants 
$A,\, B$ obtained above, it can be verified that $z(\infty) = e^{2 \pi i/3}$.
We note here that $S(0, \half, \sf13; z)$ is the reciprocal of the well known
$J$ invariant associated with the modular group $\GN(1)$.

In order to derive the implicit parametrization $y(z(s))$ for the Chazy solution,
one needs to calculate the Wronskian $W(\chi_1,\, \chi_2)$ of the two solutions
specified in \eqref{z-case1}, of the hypergeometric equation \eqref{hyper-case1}.
A short calculation employing the analytic continuation formula for 
$_2F_1(\sf{1}{12},\sf{5}{12};\half;1-s)$ near $s=0$ gives
$$ W(\chi_1,\, \chi_2) = \frac{i}{2 \pi}[W(\chi_1,\,h(s)) -\frac{\chi_1^2}{s}] =
Cs^{-1}(s-1)^{-\half} \,,$$
where $h(s)$ is analytic near $s=0$, and the last expression on the right
follows from Abel's formula. Hence, letting $s \to 0^+$ in above, yields the 
constant $C=1/2\pi$. Then from \eqref{yimp} with $b_1=\half$ and $b_2=\sf56$ 
(cf. Table 1), the following parametrization is obtained
\begin{equation}
y(z(s)) = \pi i (1-s)^{\half}(\chi_1^2 + 12s\chi_1\chi_1')\,, \qquad
\chi_1(s) = \,_2F_1(\sf{1}{12},\sf{5}{12};1;s)\,.
\label{yimp-case1}
\end{equation}
Recall from Section 1 that Ramanujan's modular function $P(q)$ introduced
in \eqref{PQR} is related to the Chazy solution via $y(z) = \pi i P(q)$.
Therefore, from \eqref{yimp-case1} it is possible to obtain an implicit
parametrization for Ramanujan's $P, Q, R$ in terms of the solutions of
the hypergeometric equation \eqref{hyper-case1}.
\begin{proposition} \label{Ppara-1}
Let $z(s)$ be the quotient of hypergeometric solutions $\chi_1$ and $\chi_2$ 
defined in \eqref{z-case1}, and $q = e^{2 \pi i z(s)}$, then
$$
P(q) = (1-s)^{\half}(\chi_1^2 + 12s\chi_1\chi_1')\,, \qquad
Q(q) = \chi_1^4\,, \qquad R(q) = (1-s)^{\half}\chi_1^6 \,. $$
\end{proposition}
\begin{proof} 
First, from \eqref{PQR} note that the Ramanujan functions satisfy the 
conditions $P \to 1, \, Q \to 1, \, R \to 1$ as $q \to 0$. Next, from \eqref{chazy}
and using the forms $f_2, f_3$ from Lemma \ref{fk} it is easy to verify that the triple 
$\{y/\pi, \, \sf{6}{\pi^2}f_2,\,\sf{9}{(i\pi)^3}f_3\}$ satisfies Ramanujan's differential
system \eqref{dPQR}. On the other hand, it follows from \eqref{yimp-case1} that
$y/\pi i \to 1$ as $s \to 0$ (equivalently, $q \to 0$). Moreover, differentiating
$y(z)$ in \eqref{yimp-case1} successively and using $s'(z) = \chi_1^2/W(\chi_1, \chi_2)$,
one obtains the expressions $\sf{6}{\pi^2}f_2 = \chi_1^4$ and 
$\sf{9}{(i\pi)^3}f_3 = \sqrt{1-s}\,\chi_1^6$, which satisfy the {\it same} 
conditions as $Q$ and $R$, respectively when $q \to 0$. Therefore, uniqueness of solutions
of the differential system \eqref{dPQR} yields the desired result.  
\end{proof}
\noindent Substituting the $q$-expansion for $S(0, \half, \sf13; z)$ into 
the hypergeometric series for $\chi_1$, it is possible to recover from the above
parametrizations the $q$-series for $P, Q, R$ in \eqref{PQR}.
On the other hand, Proposition~\ref{Ppara-1} provides an elegant representation for the
the triangle function $S(0, \half, \sf13; z)$ as well as a remarkable identity, namely
$$ S(0, \half, \sf13; z) = \frac{Q^3-R^2}{Q^3} \,, \qquad
Q^{\sf14} = \,_2F_1(\sf{1}{12},\sf{5}{12};1; \sf{Q^3-R^2}{Q^3}) \,.$$
The first expression can be used together with \eqref{PQR} to obtain the $q$-expansion
for $S(0, \half, \sf13; z)$ derived above. 

The transformation property \eqref{ytransf} for $y(z)$ under the action
of the modular group $\GN = \GN(1)$ implies that $P = y/\pi i$ is a 
{\it quasi-modular} form of weight 2 and affinity coefficient 
$p=6$ on $\GN(1)$. Moreover, by comparing
\eqref{dPQR} with the result of Lemma~\ref{fk}, it follows that $Q = \sf{6}{\pi^2}f_2$
and $R = \sf{9}{(\pi i)^3}f_3$ are {\it modular} forms of weight 4 and 6 respectively,  
on $\GN(1)$. In fact, $P=E_2, \, Q=E_4$ and $R=E_6$, where $E_k(q)$ is the 
normalized Eisenstein series of weight $k$ for the modular group $\GN(1)$. $E_k$
is a holomorphic modular form, and is defined (for even positive integer $k$) as 
\begin{equation}
E_k(q)=1-\frac{2k}{B_k} \sum_{n=1}^{\infty}\sigma_{k-1}(n)q^n\,, \qquad
q = e^{2 \pi i z}\,, \,\, z \in \UH\,,
\label{Ek}
\end{equation}
where $B_k$ is the $k$th Bernoulli number and $\sigma_k(n)$ is the sum-of-divisor
function introduced in Section 1. Detailed discussions of modular forms
appear in several monographs (see e.g.~\cite{Rankin1,Schoeneberg}.
The vector space of holomorphic modular forms
of weight $k$ on $\GN(1)$ is denoted by $M_k(\GN(1))$. Proposition~\ref{Ppara-1} 
provides a parametrization for any $f \in M_k(\GN(1))$ since it is known from
the theory of modular forms that $f$ belongs to the polynomial ring $\C [E_4,\,E_6]$.
For $4 \leq k \leq 10$, 
$M_k(\GN(1))$ is one-dimensional, and is spanned by the Eisenstein series $E_k$.
It turns out that the left hand side of \eqref{chazy} can be expressed as
$f_4+4f_2^2$ (see Section 3) which by virtue of Lemma ~\ref{fk}, must be in 
$M_8(\GN(1))$. Hence, $f_4+4f_2^2 = \C E_8$. But with $y = \pi i E_2$, the
expression $f_4+4f_2^2$ is a differential polynomial in $E_2$, and vanishes as 
$q \to 0$, whereas $E_8 \to 1$. Therefore, the Chazy equation $f_4+4f_2^2=0$ follows 
from this ``modular'' argument. 
\begin{center}
{\it 4.2. Case 2}
\end{center}
The triangle group associated with $\alpha = 0,\, 
\beta=\sf13,\, \gamma=\sf13$ is denoted by $\GN^2$. It is a normal subgroup of the 
modular group $\GN(1)$ of index 2, generated by the period-3 transformations
$$ T_1: \quad z \to \frac{-1}{z+1}\,, \qquad \quad T_2: \quad z \to \frac{z-1}{z} \,,$$
whose respective fixed points $e^{2 \pi i/3}$ and $e^{\pi i/3}$ in $\UH$ form the 
vertices $z(1)$ and $z(\infty)$ of the fundamental triangle T of $\GN^2$ together 
with the parabolic fixed point $z(0) = i\infty$. An important point to note here
is that the stabilizer of the vertex $z(0) = i\infty$ in $\GN^2$ is given by
the transformation $T_2T_1:\,\, z \to z+2$ instead of the translation $z \to z+1$, which
is {\it not} an element of $\GN^2$.

One proceeds with the
construction of the triangle function $S(0, \sf13, \sf13; z)$ in a similar manner as 
in Case 1. Define the conformal mapping by 
\begin{equation}
z(s)= \frac{\chi_2}{\chi_1}\,, \qquad \chi_1 =\, _2F_1(\sf{1}{6}, \sf{1}{2};1;s)\,,
\quad \chi_2 = A\,\,_2F_1(\sf{1}{6}, \sf{1}{2};\sf23;1-s) + B\chi_1\,,
\label{z-case2}
\end{equation}
where $_2F_1(\sf{1}{6}, \sf{1}{2};1;s),\,_2F_1(\sf{1}{6}, \sf{1}{2};\sf23;1-s)$
are solutions holomorphic in the neighborhoods of $s=0$ and $s=1$ respectively, of
the hypergeometric equation
\begin{equation}
\chi'' + \left(\frac{1}{s}+\frac{2}{3(s-1)}\right)\chi' +
\frac{1/12}{s(s-1)}\chi = 0 \,,
\label{hyper-case2}
\end{equation}
with parameters $a=\sf{1}{6}, \, b=\sf{1}{2}$ and $c=1$.  
The constants $A, B$ are determined as before by considering the analytic continuation of
the solutions near $s=0$, where
$$ _2F_1(\sf{1}{6}, \sf{1}{2}; \sf23;1-s) = \frac{\GN(\sf23)}{\GN(\sf{1}{6})\GN(\sf{1}{2})}
\left[-\log(s)_2F_1(\sf{1}{6},\sf{1}{2};1;s) +
\sum_{n=0}^{\infty}\frac{(\sf{1}{6})_n(\sf{1}{2})_n}{(n!)^2}\,h_ns^n \right]\,,
$$
with $h_n:=2\psi(1+n)-\psi(\sf{1}{6}+n)-\psi(\sf{1}{2}+n)$. Then from \eqref{z-case2}, 
one has
$$ z(s) = A_1(-\log(s) + h(s)) + B\,, \qquad
A_1 = A\,\frac{\GN(\sf23)}{\GN(\sf{1}{6})\GN(\sf{1}{2})}\,, $$
where $h(s)$ is a holomorphic function near $s=0$. As indicated earlier, since
$z \to z+2$ stabilizes the vertex at $z(0) = i\infty$, the pair $\{\chi_1,\,\chi_2\}$
must form a basis of solutions near $s=0$ with monodromy 
$\big(\begin{smallmatrix} 1 & 2 \\ 0 & 1 \end{smallmatrix}\big)$. Thus, the constant $A$ is
determined by taking $A_1 = i/\pi$ above, so that 
$z \to z+2$ onto the next branch as $s$ makes a circuit around $s=0$. The constant
$B$ is then obtained from the condition $z(1) = e^{2 \pi i/3}$, which yields
$B = -e^{\pi i/3}$ after a similar calculation as in case 1. With these values of
$A,\, B$ in \eqref{z-case2}, it is possible to construct the power series
for $q := e^{\pi i z(s)}$ near $s=0$, and its inverse $s(z) = B_2q(1+b_1q+b_2q^2+ \ldots)$
which gives the $q$-series for the triangle function $S(0,\sf12, \sf13; z)$. 
The constant $B_2 = i48\sqrt{3}$, which follows from the leading coefficient 
$h_0 = 2\psi(1)-\psi(\sf16)-\psi(\half)$ above and evaluation of Digamma functions at 
rational arguments.

From the Wronskian $W(\chi_1,\,\chi_2)$ of the two solutions of \eqref{hyper-case2},
one computes the constant $C = A_1 = i/\pi$ in \eqref{yimp}. With this value of $C$,
and $b_1=b_2=\sf23$ from Table 1, \eqref{yimp} then gives the parametrization 
\begin{equation}
y(z(s)) = \pi i P = \pi i (1-s)^{\sf23}(\chi_1^2 + 6s\chi_1\chi_1')\,, \qquad
\chi_1(s) = \,_2F_1(\sf{1}{6},\sf{1}{2};1;s)\,.
\label{yimp-case2}
\end{equation}
One also deduces from arguments aligned with Proposition \ref{Ppara-1} that, 
$$ Q(q) = (1-s)^{\sf13}\chi_1^4\,, \qquad R(q) = (1-\half s)\chi_1^6 \,, 
\qquad q = e^{2 \pi i z} \,. $$

\begin{center}
{\it 4.3. Cases 3, 4 \& 6}
\end{center}
The automorphism groups in these cases correspond to
the level-$N$ congruence subgroups $\GN_0(N)$ of the modular group $\GN(1)$, 
for $N=2,\,3,\,4$. The congruence subgroup $\GN_0(N)$ is defined by
$$\Gamma_{0}(N):=\left\{\gamma = \begin{pmatrix} a & b \\ c & d\end{pmatrix}
\in \SL : c \equiv 0 \, \pmod{N}\right\}\,.$$
The fundamental triangle T in each of these 3 cases has 2 parabolic vertices
at $z(0)=i\infty,\, z(\infty) = 0$, while the vertex $z(1)$ corresponds to
an elliptic fixed point of order 2 at $z(1) = i$ for $\GN_0(2)$, an elliptic 
fixed point of order 3 at $z(1) = e^{2 \pi i/3}$ for $\GN_0(3)$, and a parabolic 
point at $z(1) = \half$ for $\GN_0(4)$. Recall from Section 3 that in Case 6(b),
the automorphism group is $\GN(2)$ which is conjugate to $\GN_0(4)$. The 
principal congruence subgroup $\GN(2)$ is defined as
$$\Gamma(2):=\left\{\gamma \in P\SL : \gamma \equiv 
\begin{pmatrix} 1 & 0 \\ 0 & 1\end{pmatrix} \, \pmod{2}\right\}\,,$$
which is a normal subgroup of $\GN(1)$ of index 6. The fundamental
triangle of $\GN(2)$ is same as that of $\GN_0(4)$, but with respect to the variable 
$z/2$. That is, the parabolic vertices are located at $z(0) = i\infty, z(\infty) = 0$,
and $z(1) = 1$ with respect to the local coordinate $z \in \UH$.

The triangle function $S(0, \beta, 0; z)$
in each of these cases will be constructed by developing a $q$-expansion near
the parabolic vertex $z(0) = i\infty$, which is stabilized by the translation
$z \to z+1$. From \eqref{hyper-1}, the hypergeometric equation corresponding to
exponent differences $0, \beta, 0$ is
\begin{equation}
\chi''+\Big(\frac{1}{s}+\frac{1-\beta}{s-1}\Big)\chi'+
\frac{(\beta-1)^2}{4s(s-1)}\chi=0 \,,
\label{hyper-case3}
\end{equation}
with parameters $a=b=(1-\beta)/2$, and c=1. So the exponents are $0,\,0$ at the regular 
singular point $s=0$, and $a,\,a$ at $s=\infty$. A one-dimensional space
of solutions for \eqref{hyper-case3} near $s=0$ and $s=\infty$ respectively, are
spanned by the functions $_2F_1(a, a;1;s)$ and $(-s)^{-a}\,_2F_1(a,a;1;s^{-1})$.
The second linearly independent solution near each singular point contains
logarithms. The conformal mapping is defined as
\begin{equation}
z(s)= A\frac{\chi_2}{\chi_1}\,, \qquad \chi_1 =\, _2F_1(a,a;1;s)\,,
\quad \chi_2 = (-s)^{-a}\,\,_2F_1(a, a;1;s^{-1}) \,,
\label{z-case3}
\end{equation}
where the constant $A$ is determined by considering the analytic
continuation of $\chi_2$ near $s=0$. This follows from the formula
(see e.g.~\cite{Bateman}),
$$ _2F_1(a, a; 1;s^{-1}) = \frac{\sin (\pi a)}{\pi}(-s)^{a}
\left[-\log(-s)_2F_1(a,a;1;s) +
\sum_{n=0}^{\infty}\frac{(a)_n^2}{(n!)^2}\,h_ns^n \right]\,,
$$
$-\pi < \arg(-s) < \pi$, and $h_n:=2\psi(1+n)-\psi(a+n)-\psi(1-a-n)$. Proceeding 
analogously as in
the previous two cases, $z(s)$ in \eqref{z-case2} takes the form
$$ z(s) = A_1(-\log(-s) + h(s)) \,, \qquad
A_1 = A\,\frac{\sin(\pi a)}{\pi}\,, $$
and $h(s)$ is holomorphic near $s=0$. The constant $A$ (see Table 2) is
determined by taking $A_1 = i/2\pi$ above for cases 3, 4, and 6a, such 
that $\{\chi_1,\,\chi_2\}$
has monodromy $\big(\begin{smallmatrix} 1 & 1 \\ 0 & 1 \end{smallmatrix}\big)$ 
at $s=0$. Finally, inverting the power series for $q := e^{2\pi i z(s)}$ leads
to the expansion $s(z) = \widehat{B}q(1+b_1q+b_2q^2+ \ldots)$ for $S(0,\beta, 0; z)$.
The value of the constant $\widehat{B}=-e^{-h_0}$ with $h_0 = 2\psi(1)-\psi(a)-\psi(1-a)$,
is listed in Table 2 below for each group $\GN_0(N)$. For case 6b corresponding
to the group $\GN(2)$, the parabolic vertex $z(0) = i\infty$ is stabilized by the
translation $z \to z+2$. Thus one needs to choose $A_1 = i/\pi$ above, which leads
to $A = i$ and $q = e^{\pi i z(s)}$ in this case.  
 
\begin{table}[h!]
\begin{center}
\begin{tabular}{|c|c|c|c|c|} \hline
$N$ & $\beta$ & $a = (1-\beta)/2$ & $A=i \csc(\pi a)/2$ & $\widehat{B}$ \\ \hline
2 & $\half$ & $\sf{1}{4}$ & $i/\sqrt{2}$ & $-64$ \\ \hline
3 & $\sf{1}{3}$ & $\sf{1}{3}$ & $i/\sqrt{3}$ & $-27$ \\ \hline
4 & $0$ & $\sf{1}{2}$ & $i/2$ & $-16$ \\ \hline
\end{tabular}
\end{center}
\kern-\medskipamount
\caption{parameters for constructing the triangle functions $S(0, \beta, 0; z)$.}
\label{t:case3}
\end{table}
We remark that the triangle functions $S(0, \beta, 0; z)$ associated with 
$\GN_0(N),\, N=2,3,4$ as well as $\GN(2)$ can be expressed compactly in terms of Dedekind 
eta functions. 
In fact, such expressions are available for the automorphic functions associated with 
all subgroups $\GN_0(N)$ which are of genus zero, see Ref~\cite{Maier1} for a complete list.

For the cases 3, 4 and 6a, the Wronskian of the two solutions in \eqref{z-case3} is given by 
$W(\chi_1,\,\chi_2) = Cs^{-1}(s-1)^{-2a}$, where $C=(-1)^{2a}/2 \pi i$. Then from 
\eqref{yimp} one obtains the parametrization,
\begin{equation}
y(z(s)) = -6\pi i (1-s)^{2a-1}[\{(1+2a-b_1-b_2)s+(b_2-1)\}\chi_1^2+2s(s-1)\chi_1\chi_1']\,, 
\label{yimp-case3}
\end{equation}
where $\chi_1(s) = \,_2F_1(a,a;1;s)$ and $(b_1, b_2)$ is given in Table 1. Note
that in case 6b, the constant $C=i/\pi$ in the Wronskian function. The 
parametrization of the Ramanujan functions $P, Q, R$ for all three cases are 
listed in Table 3. These are obtained in the same manner as outlined
in Proposition~\ref{Ppara-1}.
Note that there are two distinct parametrizations for each of
the three cases corresponding to $\GN_0(N),\, N=2,3,4$. For $\GN_0(2)$ or $\GN_0(3)$, 
the pair of parametrizations are related via the involution $s \to s^{-1}$, as noted
below Table 1. Under this involution, it follows from \eqref{z-case3} that 
$z \to \sf{A^2}{z}$ as the two linearly independent solutions $\chi_1$ and $\chi_2$ 
of \eqref{hyper-case3} are switched in the quotient $z$. This is related
(under appropriate normalization) to the Fricke involution: $z \to -1/Nz$ for $\GN_0(N)$
(see~\cite{Maier1} for details). Thus in Table 3, the two parametrization in the case
of $\GN_0(2)$ or $\GN_0(3)$ are transformed from one to the other by switching
$\chi = \,_2F_1(a,a;1;s)$ to $\chi = (-s)^{-a}\,_2F_1(a,a;1;s)$, $s \to s^{-1}$ as well
as, taking into account the sign reversal in the Wronskian $W$ appearing in \eqref{yimp}
through the chain rule formula $d/dz = s'(z)d/ds = (\chi^2/W)d/ds$.

\begin{table}[h]
\begin{center}
\begin{tabular}{|p{1cm}|p{1cm}|c|c|c|} \hline
$\GN$ & $(b_1, b_2)$ & $P=y/\pi i$ & $Q$ & $R$ \\ \hline
$\GN_0(2)$ & $(\half, \sf{2}{3})$ & $2(1-s)^{1/2}(\chi^2+6s\chi\chi')$
& $(4-s)\chi^4$ & $(1-s)^{1/2}(s+8)\chi^6$ \\ \hline
$\GN_0(2)$ & $(\half, \sf{5}{6})$ & $(1-s)^{1/2}(\chi^2+12s\chi\chi')$
& $(1-4s)\chi^4$ & $(1-s)^{1/2}(8s+1)\chi^6$ \\ \hline
$\GN_0(3)$ & $(\sf{2}{3}, \half)$ & $3(s-1)^{2/3}(\chi^2+4s\chi\chi')$
& $(s-1)^{1/3}(s-9)\chi^4$ & $(27-18s-s^2)\chi^6$ \\ \hline
$\GN_0(3)$ & $(\sf{2}{3}, \sf{5}{6})$ & $(s-1)^{2/3}(\chi^2+12s\chi\chi')$
& $(s-1)^{1/3}(9s-1)\chi^4$ & $(1+18s-27s^2)\chi^6$ \\ \hline
$\GN_0(4)$ & $(\sf{5}{6}, \sf{5}{6})$ & $(1-2s)\chi^2-12s(s-1)\chi\chi'$
& $16(s^2-s+\sf{1}{16})\chi^4$ & $32(2s-1)(s^2-s-\sf{1}{32})\chi^6$ \\ \hline
$\GN(2)$ & $(\sf{2}{3}, \sf{2}{3})$ & $(1-2s)\chi^2-6s(s-1)\chi\chi'$
& $(s^2-s+1))\chi^4$ & $\half(2s-1)(s^2-s-2)\chi^6$ \\ \hline
\end{tabular}
\end{center}
\kern-\medskipamount
\caption{Parametrization of $P, Q, R$. Here $\chi=\,_2F_1(a,a;1;s)$ and
$s(z) = S(0, \beta, 0; z)$ with $a$ and $\beta$ values given in Table 2.}
\label{t:PQR}
\end{table}

The relationship between the parametrization of the Ramanujan functions
(Eisenstein series) $P, Q, R$ listed in Table 3 and Ramanujan's theories
of elliptic integrals will now be established. As mentioned in Section 1, Ramanujan
originally gave the parametrization \eqref{Ppara} in terms of complete elliptic
integrals of the first kind. These correspond to the parametrization associated 
to the groups $\GN_0(4)$ and $\GN(2)$ in Table 3. The parametrization corresponding to
$\GN_0(2)$ and $\GN_0(3)$ are related to the parametrization \eqref{signature}
in Ramanujan's alternative theories. The hypergeometric functions appearing in 
\eqref{Ppara} and \eqref{signature} are related to those in Table 3 via the 
well-known Pfaff transformation
\begin{equation*}
_2F_1(a,b;c;s) = (1-s)^{-a}\,_2F_1(a,c-b;c;x(s))\,, \qquad \quad x(s) = \frac{s}{s-1}\,,
\label{pfaff-1}
\end{equation*}
so that in the present case with $a=b$, $c=1$, \eqref{z-case3} takes the form
\begin{equation}
z(s) = A(-s)^{-a}\frac{_2F_1(a,a;1;s^{-1})}{_2F_1(a,a;1;s)}
 = \frac{i}{2\sin(\pi a)} \frac{_2F_1(a,1-a;1;1-x)}{_2F_1(a,1-a;1;x)}\,.
\label{z-case3a}
\end{equation}
Then by setting $a=1/r$ and $2 \pi i z(s) = -u_r$, \eqref{z-case3a} coincides 
with \eqref{signature}, and thus one recovers from Table 3 the Eisenstein series 
parametrizations in Ramanujan's theories for signatures $r=4,3,2$ corresponding to
the groups $\GN_0(2), \GN_0(3)$ and $\GN_0(4)$, respectively. For completeness, these
are presented in Table 4 below, in terms of the hypergeometric function 
$\chi_r(x):=\, _2F_1(\sf{1}{r},\sf{r-1}{r};1;x)$ which appears in Ramanujan's theories.

\begin{center}
{\it 4.4. Case 5}
\end{center}
The triangle group corresponding to $\alpha=0,\,\beta = \sf{2}{3},\,\gamma=0$
is not a Fuchsian group of first kind ($1/\beta$ is not a positive integer) although the 
triangle functions $s(z) = S(0, \half, \sf{1}{3}; z)$ and 
$\widehat{s}(z) = S(0, \sf{2}{3}, 0; z)$ are related via 
$s(\epsilon z) = -4\widehat{s}(z)/(1-\widehat{s}(z))^2, \, \epsilon = -\sqrt[3]{\sf14}$,
as mentioned in Section 3.    
The main significance of this case lies in the fact that it corresponds to Ramanujan's
alternative theory of signature $r=6$ as outlined below. 

The triangle function $S(0, \sf{2}{3}, 0; z)$ can
be constructed employing the hypergeometric theory in exactly the same way as for the
cases associated with the subgroups $\GN_0(N),\, N=2, 3, 4$ of the modular group,
even though this case is not known to have any ``modular interpretation''~\cite{Maier1}.
Specifically, one proceeds from the hypergeometric equation ~\eqref{hyper-case3} with
$\beta = \sf{2}{3}$, and defines $z(s)$ as in \eqref{z-case3} with $a=\sf{1}{6}$.
Then one finds that the constants $A = i$ and $\widehat{B} = -e^{h_0} = -432$. The
Wronskian in this case turns out to be $W(\chi_1, \chi_2) = C s^{-1}(s-1)^{-1/3}$ 
where $C=-1/2 \pi i$. Finally, substituting $(b_1, b_2) = (\sf{1}{3}, \sf{5}{6})$ from 
Table 1 into \eqref{yimp-case3} and using arguments similar to that in
Proposition~\ref{Ppara-1}, one obtains the following 
parametrization for the Eisenstein series in Ramanujan's theory of signature 
$r=a^{-1} = 6$,
\begin{gather}
P(q) = \frac{y(z(s))}{\pi i} = (1-s)^{\sf13}(\chi_1^2 + 12s\chi\chi')\,, \qquad
Q(q) = (1-s)^{\sf23}\chi^4\,, \qquad R(q) = (1+s)\chi^6 \,, \nonumber \\
\chi(s) = \,_2F_1(\sf{1}{6},\sf{1}{6};1;s) \qquad 
q = \exp\Big(-2 \pi (-s)^{-\sf16} \,\frac{_2F_1(\sf{1}{6},\sf{1}{6};1;s^{-1})}
{_2F_1(\sf{1}{6},\sf{1}{6};1;s)}\Big)\,.
\label{yimp-case5}
\end{gather}
The map in \eqref{z-case3a} induced by the Pfaff transformation also applies to this
case, and leads to the alternative parametrization for $r=6$ listed in Table 4.

\begin{table}[h]
\begin{center}
\begin{tabular}{|p{6mm}|p{1cm}|c|c|c|} \hline
$r$ & $(b_1, b_2)$ & $P=y/\pi i$ & $Q$ & $R$ \\ \hline
$2$ & $(\sf56, \sf56)$ & $(1-5x)\chi_2^2 + 12x(1-x)\chi_2\chi_2'$
& $(1+14x+x^2)\chi_2^4$ & $(1+x)(1-34x+x^2)\chi_2^6$ \\ \hline
$2*$ & $(\sf23, \sf23)$ & $(1-2x)\chi_2^2+6x(1-x)\chi_2\chi_2'$
& $(1-x+x^2)\chi_2^4$ & $(1+x)(1-\sf52 x + x^2)\chi_2^6$ \\ \hline
$3$ & $(\sf{2}{3}, \half)$ & $(3-4x)\chi_3^2+12x(1-x)\chi_3\chi_3'$
& $(9-8x)\chi_3^4$ & $(8x^2-36x+27)\chi_3^6$ \\ \hline
$3$ & $(\sf{2}{3}, \sf{5}{6})$ & $(1-4x)\chi_3^2+12x(1-x)\chi_3\chi_3'$
& $(1+8x)\chi_3^4$ & $(1-20x-8x^2)\chi_3^6$ \\ \hline
$4$ & $(\sf12, \sf23)$ & $(2-3x)\chi_4^2+12x(1-x)\chi_4\chi_4'$
& $(4-3x)\chi_4^4$ & $(8-9x)\chi_4^6$ \\ \hline
$4$ & $(\sf12, \sf56)$ & $(1-3x)\chi_4^2+12x(1-x)\chi_4\chi_4'$
& $(1+3x)\chi_4^4$ & $(1-9x)\chi_4^6$ \\ \hline
$6$ & $(\sf13, \sf56)$ & $(1-2x)\chi_6^2+12x(1-x)\chi_6\chi_6'$
& $\chi_6^4$ & $(1-2x)\chi_6^6$ \\ \hline
\end{tabular}
\end{center}
\kern-\medskipamount
\caption{Parametrization of $P, Q, R$ via $\chi_r=\,_2F_1(\sf1r,\sf{r-1}{r};1;x)$,
$x(s) = s/(s-1)$. The corresponding triangle functions are $s(z)=S(0,\beta,0; z)$ 
with $\beta = 1-\sf2r$.}
\label{t:PQRr}
\end{table}

The parametrizations listed in Table 4 can be found elsewhere, e.g., in
Refs.~\cite{BBG95,Berndt1}.
The first entry in Table 4 for $r=2$ appears in \eqref{Ppara} as Ramanujan's original
parametrization.
The second entry (case 2*) corresponds to the parametrization of the Eisenstein
series in terms of the elliptic modular function $x=\lambda(z)$. Recall from Table 1
that the triangle function for $\GN(2)$ is $S(0, 0, 0; z/2)$ where $S(0, 0, 0; z)$ is the
triangle function for $\GN_0(4)$. Then from the formula given at the end of Section 3, it
follows that $S(0, 0, 0; z/2) = \lambda(z)[\lambda(z)-1]^{-1}$.
The expression of $x(s)$ in the Pfaff transformation formula implies
that the involution $s \to \sf1s$ corresponds to the involution $x \to 1-x$, 
as well as the switching of the two $_2F_1$ functions
in the quotient $z(s)$ in \eqref{z-case3a}. This leads to the similar transformation as 
in Table 3, between the two cases for $r=3$ or 4 listed in Table 4. It is interesting to
note that the parametrization in the sextic case remain invariant under the involution 
$x \to 1-x$.

\begin{center}
{\bf 5. Concluding remarks}
\end{center}
In this note, we have reviewed the relationship between the Chazy equation
and the conformal maps defined by ratios of solutions of Fuchsian equations with 3 regular 
singular points. The Chazy equation turns
out to be a particular case of more general third order nonlinear differential
equations which can be derived systematically by developing the transformation
properties of functions under the projective monodromy group associated with the 
Fuchsian equation. Much of this note has been devoted to exemplify the important connection
between the Chazy equation and Ramanujan's study of elliptic integrals and theta functions
which play a defining role in the contemporary theory of modular forms and elliptic surfaces.
By the way of elucidating this beautiful relationship, we have derived all possible
parametrizations of the Chazy solution $y(z)$ via the Schwarz triangle functions.
We show that these parametrizations are also related to Ramanujan's parametrization for 
the Eisenstein series $P, Q, R$ in terms of hypergeometric functions, arising
in his theories of modular equations of signatures 2, 3, 4 and 6. We
also give two additional hypergeometric parametrizations stemming from the
triangle functions associated with the modular group $\GN(1)$ and its index 2 subgroup
$\GN^2$. The case corresponding to the full modular group $\GN(1)$ was considered
by Chazy himself \cite{Chazy1}. Furthermore, we note that it is possible to systematically
construct a number of third order nonlinear equations associated with modular
as well as other automorphic groups. Some of these equations corresponding to
the subgroups groups $\GN_0(2), \GN_0(3)$ and $\GN_0(4)$ of $\GN(1)$ have been studied 
recently in Refs.~\cite{ACH2006,Maier2} and in Ref.~\cite{Huber} which
analyzes the Ramanujan's differential system \eqref{dPQR} as mentioned in Section 1.
We end this note with an application of the parametrization
of Eisenstein series in Table 4 to solutions of differential equations.
 
Consider the DH system of differential equations which was introduced in Section 2, and
which corresponds to the system \eqref{gdh} with $\tau = 0$. As mentioned there, if
the variables $w_j, \, j=1, 2, 3$ satisfy the DH system then $y(z) := 2(w_1+w_2+w_3)$ 
solves the Chazy equation \eqref{chazy}. In the process of verifying this assertion, 
one computes $y'(z) = 2(w_1w_2+w_2w_3+w_3w_1)$ and $y''(z) = 12w_1w_2w_3$, which implies
that the DH variables $w_j$ are the distinct roots of the cubic $w^3 - (y/2)w^2 + (y'/2)w-y''/12 = 0$.  
Expressing $y, y', y''$ in terms of the functions $P, Q, R$ (see e.g. Proposition~\ref{Ppara-1})
and introducing $t=(6/i \pi)w$, one obtains the cubic equation
$$ t^3 - 3Pt^2 + 3(P-Q^2)t - (P^3-3PQ+2R) = 0 \,,$$
whose coefficients are polynomials in $P, Q, R$. Now note from Table 1 that 
$y(z) := 2(w_1+w_2+w_3)$ corresponds to the case 6b associated with the group $\GN(2)$.
The hypergeometric parametrization for this is the case r=2* in Table 4. Using those
parametrizations for $P, Q, R$ in the above cubic, one can verify that the cubic
factorizes as $(t-t_1)(t-t_2)(t-t_3)$ with
$$ t_1 = 3(1-x)\chi_2^2 + t_3\,, \qquad t_2 = -3x\chi_2^2 + t_3\,, \qquad
t_3 = 6x(1-x)\chi_2\chi_2' \,. $$
Thus one immediately recovers the special solutions $w_j = (\pi i/6)t_j, \, j=1, 2, 3$
for the DH system in terms of elliptic modular form $x = \lambda(z)$ and 
$\chi_2 = \,_2F_1(\half, \half; 1; x) = (2/\pi)K(x)$, where $K(x)$ is the complete
integral of the first kind defined in Section 1.
Such linearization procedures to derive exact solutions can be also employed to
a large class of nonlinear differential equations with automorphic properties
as discussed in this article. These problems will be addressed elsewhere.

\begin{center}
{\bf Acknowledgments}
\end{center}
MJA is supported by NSF grant No. DMS-0602151. The work of SC is
supported by NSF grant Nos. DMS-0307181 and DMS-0807404. We thank
Professor Rob Maier and Dr. Rod Halburd for useful discussions.

\end{document}